%% file: main.tex
\documentclass[conference]{IEEEtran}

\IEEEoverridecommandlockouts
\def\BibTeX{{\rm B\kern-.05em{\sc i\kern-.025em b}\kern-.08em
    T\kern-.1667em\lower.7ex\hbox{E}\kern-.125emX}}

\input{preamble}

\begin{document}

\title{
UC, Categorically: Rigorous Diagrammatic Proofs
}
\author{
Pooya~Farshim\IEEEauthorrefmark{1}, 
Martti~Karvonen\IEEEauthorrefmark{2}\IEEEauthorrefmark{3},
Andre~Knispel\IEEEauthorrefmark{4}, 
Markulf~Kohlweiss\IEEEauthorrefmark{5}\IEEEauthorrefmark{6}, 
Philip~Wadler\IEEEauthorrefmark{5}\\
Input Output,
\IEEEauthorrefmark{1}Switzerland,
\IEEEauthorrefmark{4}Germany,
\IEEEauthorrefmark{5}UK \\
\IEEEauthorrefmark{2}University College London, UK,
\IEEEauthorrefmark{3}University of Bath, UK
\IEEEauthorrefmark{6}University of Edinburgh, UK
\thanks{M.~Karvonen was supported by the Engineering and Physical Sciences Research Council fellowship EP/V040944/1 Resources in Computation. M.~Kohlweiss was supported by Input Output through their funding of the Edinburgh ZK-Lab.}
}
\maketitle
\input{0-abs}
\input{1-intro}
\input{2-cats}
\input{3-uc-categorical}
\input{4-uc-concrete}

\input{5-conc}

\bibliographystyle{IEEEtran}
\bibliography{refs}

\appendices
\input{app-A}

\end{document}

%% file: preamble.tex
\usepackage[english]{babel}

\usepackage{amsthm}
\usepackage{amsmath}
\usepackage{amssymb}
\usepackage{cite}
\usepackage{color}
\usepackage{comment}
\usepackage[colorlinks=true, allcolors=blue]{hyperref}
\usepackage{enumitem}
\setlist[itemize]{leftmargin=*} 
\setlist[enumerate]{leftmargin=*} 
\usepackage{string-diagrams} 
\usepackage{tikz}
\usepackage{url}
\usepackage{xspace}

\usepackage{xcolor}
\usepackage[normalem]{ulem}

\numberwithin{equation}{section}
\theoremstyle{plain}
\newtheorem{theorem}[equation]{Theorem}
\newtheorem{claim}[equation]{Claim}
\newtheorem{lemma}[equation]{Lemma}
\newtheorem{proposition}[equation]{Proposition}
\newtheorem{corollary}[equation]{Corollary}
\theoremstyle{definition}
\newtheorem{definition}[equation]{Definition}

\newtheorem{remark}[equation]{Remark}

\usetikzlibrary{matrix}
\usetikzlibrary{shapes.geometric} 
\newenvironment{pic}[1][]
{\begin{aligned}\begin{tikzpicture}[font=\tiny,#1]}
{\end{tikzpicture}\end{aligned}}

\tikzset{
state/.style={semicircle,shape border rotate=90,draw},
}

\tikzset{
costate/.style={semicircle,shape border rotate=270,draw},
}

\usetikzlibrary{fit} 
\usetikzlibrary{calc}

\newcommand{\ie}{\text{i.e.,}\xspace}
\newcommand{\eg}{\text{e.g.,}\xspace}

\newcommand{\NN}{\mathbb{N}}
\newcommand{\tuple}[1]{\mathopen{\langle}#1\mathclose{\rangle}}
\newcommand{\fst}{\mathsf{fst}}
\newcommand{\snd}{\mathsf{snd}}

\newcommand{\cat}[1]{\mathbf{#1}}

\newcommand{\id}[1][]{\mathrm{id}_{#1}}
\newcommand{\CC}{\cat{C}}

\newcommand{\DD}{\cat{D}}

\newcommand{\OO}{\Sigma} 
\newcommand{\MM}{\Gamma} 
\newcommand{\CB}{\cat{C_{bd}}} 
\newcommand{\DR}{\cat{D_{real}}} 
\newcommand{\DB}{\cat{D_{bd}}} 
\newcommand{\msg}{\mathsf{msg}}
\newcommand{\ucsro}{\mathsf{subroutine\textsf{-}output}}
\newcommand{\uci}{\mathsf{input}}
\newcommand{\ucbd}{\mathsf{backdoor}}
\newcommand{\tp}{\mathsf{tp}}
\newcommand{\E}{\mathcal{E}} 
\newcommand{\exec}{\mathsf{EXEC}}
\newcommand{\ucexec}{\mathsf{EXEC}^{\mathsf{UC}}}
\newcommand{\ucid}{\mathsf{ID}}
\newcommand{\uccom}{\mathsf{C}}
\newcommand{\conn}{\mathcal{C}}
\newcommand{\names}{\mathcal{N}} 
\newcommand{\machines}{M} 
\newcommand{\mux}{\mathsf{mux}}
\newcommand{\demux}{\mathsf{demux}}

%% file: 0-abs.tex
\begin{abstract}
Category theory is a mathematical theory of composition, widely used in logic, computing, and physics. Here we apply it to give a theory of secure composition. In particular, we provide a categorical treatment of Canetti's Universal Composability (UC) framework for systems with a static number of parties and sessions, often termed UC for static systems, yielding four benefits. 

First, we present our results graphically yet retain rigor by applying a standard categorical technique known as \emph{string diagrams}. In particular, our formulation of the composition theorem can be graphically verified with a short sequence of diagrams, while remaining translatable to equations and amenable to formal verification. 

Second, categories let us generalize so that our results extend beyond interactive Turing machines to other forms of computation, such as quantum computation or domain-specific languages. 

Third, categories help us drop some unnecessary restrictions of UC (e.g., our adversary can be a computational network rather than a single Turing machine); we prove equivalence between our variant and the usual UC, showing no expressiveness is lost. 

Finally, the categorical perspective leads us to identify and correct some minor technical oversights in the standard formulation of simple UC.

\end{abstract}


%% file: 1-intro.tex
\section{Introduction}
\label{sec:intro}

While highly influential, Universal Composability (UC)~\cite{Can01,Can20} has divided, rather than unified, industrial and academic cryptographers, alongside game-based and simulation-based provable-security experts. Beyond UC, researchers in simulation-based composable security are further divided into camps such as Constructive Cryptography~\cite{Mau11} and IITM~\cite{DBLP:conf/csfw/Kusters06}. This fragmentation is common and expected in relatively young scientific disciplines. What can be learned from unification in other areas, such as mathematics?

Category Theory is the mathematics of composition. It serves as a unifying language across mathematics and is widely used in logic, computing, and physics, not least to reveal connections between the three~\cite{baez_physics_2011}.
It is perhaps surprising that despite its unifying successes in mathematics and computer science, this mathematics of composition has not yet been comprehensively applied to secure composition.

This paper expresses a minor variant of UC for static systems~\cite[Section 2]{Can20} (a.k.a.\ simple UC, henceforward simply ``UC'') and its composition theorem in categorical terms. While this restriction is made for technical simplicity, simple UC is already fairly expressive. For example, it can model multi-party computation with an a priori known set of participants~\cite{DBLP:conf/crypto/CanettiCL15} or any protocol using a static number of idealized building blocks. 
Indeed, EasyUC~\cite{CSV19:EasyUC} works essentially in this setting.
This restriction also simplifies the treatment of global functionalities~\cite{BCH+:globalwithUC}, though they are not the focus of this work.

Our approach offers several benefits:

\emph{First}, categories let us use \emph{string diagrams}, an intuitive yet rigorous diagrammatic reasoning method developed by category theorists. Consequently, our proofs of the Composition Theorems~\ref{thm:securemapsareanSMC} and~\ref{thm:categoricalUCcomposition} are easier to follow than textual presentations, yet amenable to equational translation and computer formalization. 

\emph{Second}, categories are highly useful for generalization. To quote Mac Lane~\cite[IX.6]{maclane:categories}, a co-founder of category theory,  
\begin{quote}
    The point of these observations is not the reduction of the familiar to the unfamiliar [...] but the extension of the familiar to cover many more cases.
\end{quote}

\noindent Indeed, our proof clearly holds for models of computation beyond UC's networks of interactive Turing machines (ITMs). Specifically, it provides a template for developing UC-like theories that replace ITMs with, \eg quantum systems~\cite{Unr10} or code in a suitable domain-specific language (DSL)~\cite{
CSV19:EasyUC,
DBLP:journals/pacmpl/GancherSFSM23,
LHM19}.
Crucially, one does not need to reprove the composition theorem for these variants. To invoke the general result, one simply verifies that their network model of ``interactive computational systems'' admits a computational indistinguishability notion satisfying Definition~\ref{def:indistinguishability} (which typically relies on a model of network execution).

\emph{Third}, categories tend to ``point you in the right direction'':
\begin{itemize}
    \item There is no need for machines to have names (identities) in the context of UC: instead, it is more convenient to name the connections between machines.
    \item The environment and adversary are not restricted to a single machine; they can consist of multiple machines, just like protocols. As a result, absorbing a machine into the environment or adversary is immediate, avoiding the extra step of replacing the network with a single equivalent machine. This also sidesteps a restriction Canetti notes for single-machine environments (see Remark~\ref{rem:environmentasnetwork}).
\end{itemize}
Despite these changes, we show in Section~\ref{ssec:translatingtoUC} that our version translates directly to UC (and vice versa), letting us deduce \cite[Theorem~3]{Can20} from our results. 

\emph{Finally}, we identify and fix some technical oversights (in Section~\ref{ssec:translatingtoUC}) in the setup of the UC composition theorem. These fixes are guided by the categorical principle that sequential composition requires matching types (\ie the codomain of one matches the domain of the other), and some of the corrected conditions appear in a recent tutorial by Canetti~\cite{Can25:tutorial}.

We envision further benefits once the program has been carried out in greater detail: a DSL-based variant of UC, formalizing the widespread practice of expressing computational behavior as (pseudo)code; similar analyses of other composability frameworks, guided by a categorical axiomatization that guarantees composition ``for free''; and, ultimately, rigorous security-preserving translations between frameworks, akin to functorial versions of~\cite{RKC:embeddingUCinIITM}. We expand on these in Section~\ref{sec:conclusion}.

\subsection*{Related works} This paper advances an ongoing program of applying category theory to composable security frameworks~\cite{BK22,BK23}. However, prior literature left open the task of formally connecting these abstract treatments to established frameworks. Here, we close this gap for the most prominent paradigm, namely UC. Rather than pursuing maximal generality, our objective is to provide an abstract template for deriving ``UC-like'' theories with significantly greater tractability than the original framework allows. Consequently, we establish our main results via rigorous diagrammatic reasoning, requiring fewer categorical prerequisites. The theory developed here is not an instance of the general theory in~\cite{BK23}, nor is that theory an instance of ours. In Appendix~\ref{appendix:UC-composition} we outline a more abstract account of our theory, bringing it closer to~\cite{BK23} in spirit but still differing technically.

The use of informal diagrams as a heuristic device is common in cryptography (\eg \cite{PR08,joyofcryptography,BDF+18,BO23SSP}). Recasting them as string diagrams promotes them to a rigorous tool without losing their intuitive nature. String diagrams have already been used in quantum cryptography when analyzing particular protocols (see,  \eg~\cite{breiner:graphicaldicrypto,coecke:graphicalqkd,breiner:selftesting,hillebrand:superdense,stay:crypto,colissonetal:graphstate}), whereas here we use them in the foundations of composable security. Facilitating and popularizing string diagrams in cryptography is one of our long-term goals. 

Recent work~\cite{PKWC24:UCisRC,KPC24:robustcompilationandUC} relates UC to Robust Compilation (RC), showing that UC emulation corresponds to robust hyperproperty-preserving compilation: roughly, compiled programs retain the security properties of their source counterparts no matter what malicious code they are linked with. While RC accommodates distinct source and target languages, instantiating it to UC requires the two to coincide; we instead generalize UC to arbitrary models of computation from the outset, proving the composition theorems once at that level of generality. Moreover, whereas their diagrams serve as informal illustration, our string diagrams are rigorous proof objects. We conjecture that RC compilers correspond to security-preserving symmetric monoidal functors (cf.~\cite[Theorem~4.10]{BK23}).

In Section~\ref{sec:UCcategorical} we make extensive use of a construction $\CB$ of a category with backdoors built from a category $\CC$. An essentially similar construction appears in~\cite{VDS:leakage} (there phrased as a quotient of the $\mathrm{coPara}$ construction), with the adversarial interface read as a leak.

\paragraph*{Structure of the paper} Section~\ref{sec:introtoCT} briefly introduces categories and string diagrams. Section~\ref{sec:UCcategorical} presents our main novelties, expressing UC categorically. The main results are Theorem~\ref{thm:securemapsareanSMC}, which guarantees that security is closed under sequential and parallel composition, and Theorem~\ref{thm:categoricalUCcomposition}, which gives our categorical version of the UC composition theorem. Section~\ref{sec:UCconcrete} instantiates the theory with a model based on networks of ITMs, and translates between this and UC~\cite[Section~2]{Can20}. Section~\ref{sec:conclusion} concludes with further questions.

%% file: 2-cats.tex
\section{A Crash Course on Category Theory}\label{sec:introtoCT}

This section introduces sufficient background in Category Theory for our purposes. String diagrams are surveyed in~\cite{Sel10,PZ:stringdiagrams}. A working cryptographer might find texts such as~\cite{coecke2010categories,FS19,heunenvicary:categories} easier to learn from, as they start with some applications in mind and introduce string diagrams concurrently with the material. General references for category theory include~\cite{awodey:categorytheory,leinster:basicCT}.

\subsection{Categories}\label{ssec:cats}

In mathematics and computing we often study objects of a given kind together with maps between those objects, such as:
\begin{itemize}
    \item sets and functions between sets;
    \item groups and group homomorphisms;
    \item finite-dimensional real vector spaces (with a chosen basis) where maps from a space of dimension $n$ to a space of dimension $m$ are given by $m \times n$ matrices;
    \item types in some type system, where maps from type $S$ to type $T$ consist of a well-typed term $t$ of type $T$ with a free variable $x$ of type $S$ (in shorthand, $x : S \vdash t : T$).
\end{itemize}

Category theory provides a systematic language to capture this kind of situation. A \emph{category} $\CC$ consists of \emph{objects} and \emph{maps} between them. We denote objects $A,B,C,\dots$ and maps, also called \emph{morphisms}, by $f,g,h,\dots$. Every morphism has a \emph{source} (also called the \emph{domain}) and a \emph{target} (also called the \emph{codomain}), and we write $f \colon A\to B$ to mean that $f$ is a morphism (in $\CC$) from $A$ to $B$. 

A nice feature of categories is that they admit a pictorial syntax, known as \emph{string diagrams}, that is both intuitive and rigorous. We introduce string diagrams in tandem with our discussion of categories. In a string diagram, a morphism $f \colon A\to B$ is depicted by 
\[\begin{pic}
\node[box=0/1/0/1] (f) at (0,0) {f};
\draw (f.east) to  ++(.6, 0) node[above] {$B$};
\draw (f.west) to  ++(-.6, 0) node[above] {$A$};
\end{pic}\]
We will often omit the labels for the wires when there is no risk of confusion.

In a category, morphisms can be composed whenever it makes sense to do so. For our four examples above:
\begin{itemize}
    \item In the category $\cat{Set}$ of sets and functions, if $f : A \to B$ and $g : B \to C$, their composition is given by $g \circ f : A \to C$.
    \item In the category $\cat{Grp}$ of  groups and group homomorphisms composition is defined similarly.
    \item In the category $\cat{Mat}$ of vector spaces and matrices, composition corresponds to matrix multiplication.
    \item In the category $\cat{Trm}$ of types and terms, the composition of $x \colon S \vdash t \colon T$ with $y \colon T \vdash s \colon U$
    is given by substituting $t$ for $y$ in $s$, giving $x \colon S \vdash s[t/y] \colon U$. 
\end{itemize}

In general, in a category $\CC$, there has to be a sequential composition operation $\circ$, which, when given $f \colon A\to B$ and $g \colon B\to C$, returns $g\circ f \colon A\to C$, which in string diagrams is expressed by
\[\begin{pic}
\node[box=0/1/0/1] (a) at (0,0) {g \circ f};
\draw (a.east) to ++(.6, 0) node[above] {$C$};
\draw (a.west) to ++(-.6, 0) node[above] {$A$};
\end{pic}
\enspace := \enspace
\begin{pic}
\node[box=0/1/0/1] (f) at (0,0) {f};
\node[box=0/1/0/1] (g) at (1.5,0) {g};
\draw (f.east) to node[above] {$B$} (g.west);
\draw (f.west) to ++(-.6, 0) node[above] {$A$};
\draw (g.east) to ++(.6, 0) node[above] {$C$}; 
\end{pic}\]
To match the diagrammatic order, the notation $f;g$ is sometimes used instead of $g\circ f$. 

Moreover, for every object $A$ there has to be an identity morphism $\id[A] \colon A\to A$ on it, drawn as 
\[\begin{pic}
 \draw (0,0) node[left] {$A$} to (1,0) node[right] {$A$};
 \end{pic}\]
What makes $\id[A]$ and $\id[B]$ earn their names is that they are required to satisfy $f\circ \id[A]=f=\id[B]\circ f$ for any $f \colon A\to B$. One benefit of string diagrams is how they make the required axioms pictorially intuitive, so this would be expressed by 
\[\begin{pic}
    \node[box=0/1/0/1] (f) at (0,0) {f};
    \draw (f.east) to  ++(.5, 0) node[above] {$B$};
    \draw (f.west) to  ++(-.3, 0) node[above] {$A$} to  ++(-.7, 0) node[above] {$A$} ;
    \end{pic}
     =
    \begin{pic}
    \node[box=0/1/0/1] (f) at (0,0) {f};
    \draw (f.east) to  ++(.5, 0) node[above] {$B$};
    \draw (f.west) to  ++(-.5, 0) node[above] {$A$};
    \end{pic}
    =
    \begin{pic}
    \node[box=0/1/0/1] (f) at (0,0) {f};
    \draw (f.east) to  ++(.3, 0) node[above] {$B$} to  ++(.7, 0) node[above] {$B$};
    \draw (f.west) to  ++(-.5, 0) node[above] {$A$};
    \end{pic}
\]
More informally, this amounts to saying that the length of a wire carries no information. 

Moreover, we require composition to be associative, so that whenever we have three morphisms $f \colon A\to B,g \colon B\to C$ and $h \colon C\to D$ that could be composed in sequence, we do not need to specify which pairs we compose first, i.e., $h\circ (g\circ f)=(h\circ g)\circ f$. This means that we can omit the brackets and simply write $h\circ g\circ f$, which justifies the drawing 
\[\begin{pic}
\node[box=0/1/0/1] (f) at (0,0) {f};
\node[box=0/1/0/1] (g) at (1.5,0) {g};
\node[box=0/1/0/1] (h) at (3,0) {h};
\draw (f.east) to  (g.west);
\draw (f.west) to ++(-.6, 0) node[above] {$A$};
\draw (g.east) to (h.west);
\draw (h.east) to ++(.6, 0) node[above] {$D$};
\end{pic}\]
This concludes the definition of a category: a class of objects and morphisms between them where morphisms compose, there is an identity morphism for each object, the identity morphisms are identities under composition, and composition is associative. 

The structure of a category is already sufficiently rich to express many properties of interest. For example, one can define a morphism  $f \colon A\to B$  to be an \emph{isomorphism} if there exists a morphism $g \colon B\to A$ such that $g\circ f=\id[A]$ and $f\circ g=\id[B]$. This general notion specializes to the correct notion of ``sameness'' in each context, capturing, e.g., bijections between sets and isomorphisms of groups, vector spaces, and types.

While many examples of interest arise from sets-with-further-structure and structure-preserving maps, the definition of a category does not require the morphisms to consist of some kind of functions. To see that this is not always the right thinking, it is useful to consider a partially ordered set (poset) consisting of a set $P$ equipped with a relation $\leq$ that is reflexive, transitive, and antisymmetric. Any poset $\tuple{P,\leq}$ induces a category whose objects are given by elements of $P$, and whose morphisms are defined by setting there to be exactly one morphism $x\to y$ whenever $x\leq y$ and no morphisms $x\to y$ otherwise, with reflexivity and transitivity inducing identities and composition. In fact, we will see another example of a category where the intuition of morphisms as functions is not correct: in Section~\ref{ssec:baseforUC} we will encounter a category whose morphisms are certain networks of interactive Turing machines. 

\subsection{Monoidal and symmetric monoidal categories}\label{ssec:monoidalcats}
    In many categories of interest, there is also a way of composing objects and morphisms \emph{in parallel} and not just sequentially. For example,
\begin{itemize}
    \item Recall that if $A$ and $B$ are sets then their cartesian product $A \times B = \{\tuple{x,y} | x \in A,\, y \in B\}$ is the set of all pairs of elements with the first element in $A$ and the second in $B$. Given two functions $f \colon A\to B$ and $g \colon C\to D$, there is a function $f\times g \colon A\times C\to B\times D$ which sends $\tuple{x,y}\in A\times C$ to $\tuple{f(x),g(y)}\in B\times D$. 

    \item Given two groups $G$ and $H$, the product $G\times H$ is obtained by taking the cartesian product of the underlying sets, and performing multiplication pointwise. With this definition, the product $f\times g$ of group homomorphisms $f$ and $g$ is again a group homomorphism of product groups. 

    \item Given two vector spaces $A$ and $B$ with bases $a_1,\dots, a_m$ and $b_1,\dots ,b_n$, their tensor product $A \otimes B$ is an $mn$-dimensional vector space that has a basis built out of pairs $(a_i,b_j)$. Now, given an $m\times n$-matrix $f$ and a $p\times q$-matrix $g$, their \emph{Kronecker product} $f \otimes g$ is the $mp\times nq$-matrix essentially obtained by replacing the $(i,j)$th entry $f_{i,j}$ of $f$ with the matrix $f_{i,j} g$. 
    
    \item If a type system has product types, one can use terms $x : S \vdash t : T$ and $y : U \vdash v : V$ to obtain a term $z : S\times U \vdash \tuple{t[\fst(z)/x], v[\snd(z)/y]} : T\times V$, where $\tuple{t,v}$ builds a pair and $\fst$ and $\snd$ extract the first and second components of a pair. 
\end{itemize}
Monoidal categories capture this kind of situation. Roughly speaking, a monoidal category is a category equipped with a parallel composition operation that is associative and has a unit.
The parallel composition operation is written using the tensor product symbol $ \otimes$.  Given two morphisms $f \colon A\to B$ and $g \colon C\to D$ in $\CC$, their parallel composite is a morphism $f \otimes g \colon A \otimes C\to B \otimes D$. Diagrammatically this is denoted by juxtaposition:
\[\begin{pic}
\node[box=0/2/0/2] (a) at (0,0) {f \otimes g};
\draw (a.east.1) to  ++(.6, 0) node[right] {$B$};
\draw (a.west.1) to  ++(-.6, 0) node[left] {$A$};
\draw (a.east.2) to ++(.6, 0) node[right] {$D$};
\draw (a.west.2) to  ++(-.6, 0) node[left] {$C$};
\end{pic}
\enspace := \enspace 
\begin{pic}
\node[box=0/1/0/1, minimum height=5mm] (f) at (0,.65) {f};
\node[box=0/1/0/1, minimum height=5mm] (g) at (0,0) {g};
\draw (f.east) to ++(.6, 0) node[right] {$B$};
\draw (f.west) to  ++(-.6, 0) node[left] {$A$};
\draw (g.west) to  ++(-.6, 0) node[left] {$C$};
\draw (g.east) to  ++(.6, 0) node[right] {$D$};
\end{pic}\]

This parallel composition operation must satisfy $\id[A \otimes B]=\id[A] \otimes \id[B]$ which pictorially amounts to  
\[\begin{pic}
    \draw (0,0) node[above] {$A \otimes B$} to (1.5,0) node[above] {$A \otimes B$};
  \end{pic}
  \enspace =\enspace 
  \begin{pic}
      \draw (0,0) node[left] {$A$} to (1.5,0) node[right] {$A$};
      \draw (0,-.5) node[left] {$B$} to (1.5,-.5) node[right] {$B$};
  \end{pic}
\] 
This lets us replace a wire of type $A \otimes B$ with two wires (one for $A$ and one for $B$) whenever it is convenient to do so. Morphisms between products are not required to be products of morphisms. For example, a function $(x,y)\mapsto (f(x,y),g(x,y))$ is a product only if $f$ depends only on $x$ and $g$ only on $y$. 

An arbitrary box can have multiple input and output wires. Diagrammatically, a morphism  \[f \colon A_1 \otimes \dots \otimes A_n\to B_1 \otimes \dots \otimes B_m\] is drawn as 
\[\begin{pic}
\node[box=0/3/0/3] (a) at (0,0) {f};
\draw (a.east.1) to  ++(.6, 0) node[right] {$B_1$};
\draw (a.east.3) to ++(.6, 0) node[right] {$B_m$};
\draw (a.west.1) to  ++(-.6, 0) node[left] {$A_1$};
\draw (a.west.3) to  ++(-.6, 0) node[left] {$A_n$};
\node[left] at (a.west.2) {\dots};
\node[right] at (a.east.2) {\dots};
\end{pic}\]

We also require the \emph{interchange law}, which states that 
\begin{equation}\label{eq:interchange}
(h \otimes k)\circ (f \otimes g)=(h\circ f) \otimes (k\circ g)
\end{equation}
for all $f \colon A\to B$, $g \colon C\to D$, $h \colon B\to E$, and $k \colon D\to F$. The interchange law guarantees that the string diagram 
\[\begin{pic}
\node[box=0/1/0/1, minimum height=5mm] (f) at (0,.65) {f};
\node[box=0/1/0/1, minimum height=5mm] (g) at (0,0) {g};
\node[box=0/1/0/1, minimum height=5mm] (h) at (1.5,.65) {h};
\node[box=0/1/0/1, minimum height=5mm] (k) at (1.5,0) {k};
\draw (h.east) to  ++(.6, 0) node[above] {$E$};
\draw (f.west) to  ++(-.6, 0) node[above] {$A$};
\draw (f.east) to node[above] {$B$} (h.west);
\draw (g.west) to  ++(-.6, 0) node[below] {$C$};
\draw (g.east) to node[below] {$D$} (k.west);
\draw (k.east) to  ++(.6, 0) node[below] {$F$};
\end{pic}\] 
is unambiguous. 

In a \emph{strict} monoidal category, parallel composition is \emph{unital}. This means that there is a special object $I$, called \emph{the tensor unit}, satisfying
\[
    A \otimes I = A = I \otimes A \quad\text{and}\quad \id[I] \otimes f = f = f \otimes \id[I]
\]
on objects and morphisms, respectively. Further, parallel composition is \emph{associative}, meaning
\[
    A \otimes (B \otimes C) = (A \otimes B) \otimes C  \quad\text{and}\quad f \otimes (g \otimes h)=(f \otimes g) \otimes h~.
\] This results in the notion of a \emph{strict monoidal category}.

Unfortunately, strict monoidal categories are too strict: they rule out most examples of interest.
For instance, the cartesian product of sets is not strictly associative. Rather, it is associative up to the isomorphism $((a,b),c)\mapsto (a,(b,c))$. Loosely speaking, a \emph{monoidal category} is a category with parallel composition $ \otimes$ that is associative and unital up to isomorphism. More formally, a monoidal category is a category equipped with a parallel composition operation and associativity and unitality isomorphisms which must satisfy some equations. Technically, these should be natural isomorphisms satisfying the triangle and pentagon identities. We won't delve further into these equations; see~\cite[Section 7.8]{awodey:categorytheory} for definitions. 

It is a theorem that every monoidal category is monoidally equivalent to a strict monoidal category, which roughly speaking means that for any monoidal category there is a strict monoidal category and suitable mappings between them that translate back and forth. We neither formalize nor prove this theorem here, but use this result implicitly to pretend that all monoidal categories we work with are strict, simplifying notation.

The tensor unit $I$ is denoted by the empty diagram (as juxtaposing any diagram with the empty one results in the same diagram).  Thus, a morphism $r \colon I\to  A$ would be denoted by a box with no inputs
\[\begin{pic}
\node[box=0/0/0/1] (r) at (0,0) {r};
\draw (r.east) to  ++(.6, 0) node[above] {$A$};
\end{pic}\]

Borrowing terminology from quantum circuits~\cite{AC04,heunenvicary:categories}, morphisms from $I$ are usually called \emph{states}, and we will draw states $r \colon I\to A$ and $r \colon I\to A_1 \otimes \cdots \otimes A_n$ respectively as 
   \[ \begin{pic}
    \node[state] (r) at (0,0){$r$};
    \draw (r.east) to ++(.6,0) node[above] {$A$};
  \end{pic}\quad \text{and } \quad \begin{pic}
    \node[state,minimum width=1cm] (r) at (0,0){$r$};
    \draw (r.north east) to ++(.6,0) node[right] {$A_1$};
    \node[right] at (r.east) {\dots};
    \draw (r.south east) to ++(.6,0) node[right] {$A_n$};
  \end{pic}\]
In $\cat{Set}$ states correspond to elements of a set, in $\cat{Mat}$ to vectors, in $\cat{Grp}$ to neutral elements, and in $\cat{Trm}$ to closed terms (i.e., ones with no free variables) of a given type. 

Now, a \emph{symmetric} monoidal category (SMC) is, roughly speaking, a monoidal category where $ \otimes$ is suitably commutative. Slightly more formally, an SMC comes with specified \emph{swap} or \emph{symmetry} isomorphisms $\sigma_{A,B} \colon A \otimes B \to B \otimes A$ for all objects $A, B$, satisfying some further equations. These are more easily explained in string diagrams. To start, we draw  $\sigma_{A,B} \colon A \otimes B \to B \otimes A$ as
\[\begin{pic}[yscale=.8]
\node (botl) at (0, 0) {};
\node (topl) at (0, .8) {};
\draw (topl) to node[above] {$A$} ++(.6,0) to[in=180,out=0] ++(1.2,-.8) to node[below] {$A$} ++(.6,0);
\draw (botl) to node[below] {$B$} ++(.6,0) to[in=180,out=0]  ++(1.2,.8) to node[above] {$B$} ++(.6,0);
\end{pic}\quad  \text{instead of}\quad \begin{pic}
\node[box=0/2/0/2, minimum height=9mm] (a) at (0,0) {\sigma_{A,B}};
\draw (a.west.1) to ++(-.6, 0) node[above] {$A$};
\draw (a.west.2) to ++(-.6, 0) node[below] {$B$};
\draw (a.east.1) to ++(.6, 0) node[above] {$B$};
\draw (a.east.2) to ++(.6, 0) node[below] {$A$};
\end{pic}\] 
Given this notation, we then require
\[\begin{pic}
\node[box=0/1/0/1,minimum height=6mm] (g) at (0, .8) {g};
\node[box=0/1/0/1,minimum height=6mm] (f) at (0, 0) {f};
\draw (g.east.1) to  ++(.8,0) node[above] {$D$};
\draw (f.east.1) to  ++(.8,0) node[above] {$C$};
\draw (g.west.1) to[in=0,out=180] ++(-1.2,-.8) to ++(-.75,0) node[above] {$B$};
\draw (f.west.1) to[in=0,out=180]  ++(-1.2,.8) to  ++(-.75,0) node[above] {$A$};
\end{pic}=
\begin{pic}
\node[box=0/1/0/1,minimum height=6mm] (f) at (0, .8) {f};
\node[box=0/1/0/1,minimum height=6mm] (g) at (0, 0) {g};
\draw (f.west.1) to ++(-.8,0) node[above] {$A$};
\draw (g.west.1) to ++(-.8,0) node[above] {$B$} ;
\draw (f.east.1) to[in=180,out=0] ++(1.2,-.8) to  ++(.75,0) node[above] {$C$};
\draw (g.east.1) to[in=180,out=0]  ++(1.2,.8) to  ++(.75,0) node[above] {$D$};
\end{pic}\]
or equivalently,
 $(g \otimes f)\circ  \sigma_{A,B} = \sigma_{C,D} \circ  (f \otimes g)$ 
holds for all $f \colon A\to C,g \colon B\to D$. Intuitively, this states that boxes can be slid along wires. (Formally this property corresponds to the family $\sigma$ being a \emph{natural transformation}, a concept that shows up often in category theory.)  Second, we require that crossing the same wires twice amounts to no change:
\[\begin{pic}[yscale=.8]\node (botl) at (0, 0) {};
\node (topl) at (0, .8) {};
\draw (topl) to node[above] {$A$} ++(.6,0) to[in=180,out=0] ++(1.2,-.8) to node[below] {$A$} ++(.6,0) to[in=180,out=0]  ++(1.2,.8) to node[above] {$A$} ++(.6,0);
\draw (botl) to node[below] {$B$} ++(.6,0) to[in=180,out=0]  ++(1.2,.8) to node[above] {$B$} ++(.6,0)  to[in=180,out=0] ++(1.25,-.8) to node[below] {$B$} ++(.6,0);
\end{pic}
\enspace=\enspace \begin{pic}[yscale=.8]\node (botl) at (0, 0) {};
\node (topl) at (0, .8) {};
\draw (topl) to node[above] {$A$} ++(.6,0) to ++(1,0);
\draw (botl) to node[below] {$B$} ++(.6,0) to ++(1,0);
\end{pic}
\]
or equivalently $\sigma_{B,A} \circ \sigma_{A,B} = \id[A \otimes B]$. The slogan ``only connectivity matters'' is often used to summarize the pleasant properties of string diagrams for SMCs. While we work somewhat informally with string diagrams, we  stress that they are a perfectly rigorous tool to capture exactly what follows from the axioms of an SMC. To this effect, we recall Selinger's formulation~\cite[Theorem~3.12]{Sel10} of the Joyal--Street coherence theorem~\cite[Theorem~2.3]{joyalstreet}.

\begin{theorem}[Joyal--Street~{\cite[Theorem~2.3]{joyalstreet}}, as formulated in {\cite[Theorem~3.12]{Sel10}}]
   A well-formed equation between morphisms in the language of symmetric monoidal categories follows from the axioms of symmetric monoidal categories if and only if it holds, up to isomorphism of diagrams, in the graphical language.
\end{theorem}

Given an SMC $\CC$, by a \emph{sub-SMC} of $\CC$ we mean an SMC $\DD$ whose objects and morphisms are sub-collections of those of $\CC$, and whose SMC-structure (sequential and parallel composition, identity morphisms, unit object, symmetries) coincides with that of $\CC$. Sub-SMCs are closed under arbitrary intersections: as a result, one can always form the smallest sub-SMC containing given collections of morphisms and objects, and we will refer to the end-result as the \emph{sub-SMC generated by} this data. 

\subsection{Building an SMC from generators and equations}\label{ssec:gensandrels}

We will now review how to build a symmetric monoidal category from generators and equations, as this will be used in Section~\ref{ssec:baseforUC} to build a category of networks of interactive Turing machines to model UC. A useful analogy is that of building a group from generators and equations. A more detailed yet introductory account can be found in~\cite[Chapter 2]{PZ:stringdiagrams}. 
To generate an SMC, one needs
\begin{itemize}
    \item A set $\OO$ of generating objects, a general object then being an element of the set $\OO^*$ of words in the alphabet $\OO$.
    \item For each $V,W\in \OO^*$, a (possibly empty) set $\MM_{V,W}$ of generating morphisms $V\to W$. 
\end{itemize} 

The symmetric monoidal category generated by $(\OO,\MM)$ has $\OO^*$ as its set of objects, and its morphisms can be thought of as given by string diagrams built from identity wires, symmetries and generating morphisms from the sets $\MM_{V,W}$ modulo the equations we've asserted for SMCs in the previous section (associativity and unitality of $ \otimes$ and the equations for swap).

As a warm up for using this construction in later sections, let us consider an example. We let $\OO$ consist of a single element, so that an element of $\OO^*$ is uniquely specified by its length and hence $\OO^*$ can be identified with $\NN$. We let $\MM_{0,1}$ and $\MM_{2,1}$ be singletons and draw their elements as
\[ 
   \begin{pic}[scale=.6]
    \node[dot] (r) at (0,0){};
    \draw (r) to ++(1,0);
  \end{pic}\enspace \text{and} \enspace 
  \begin{pic}[scale=.6]
    \node[dot] (r) at (0,0){};
    \draw (r) to ++(1,0);
    \draw (r) to[out=90,in=0] ++(-1,.5);
    \draw (r) to[out=270,in=0] ++(-1,-.5);
  \end{pic}
  \]
and let $\MM_{i,j}$ be empty for all other choices of $i,j\in\NN$. Then in the resulting SMC, 
 \[ \begin{pic}[scale=.6]
    \node[dot] (r) at (0,0){};
    \draw (r) to ++(1,0);
    \draw (r) to[out=90,in=0] ++(-1,.5);
    \draw (r) to[out=270,in=0] ++(-1,-.5);
    \node[dot] (s) at (0,1){};
    \draw (s) to ++(1,0);
  \end{pic}
  \enspace ,  \enspace 
   \begin{pic}[scale=.6]
    \node[dot] (r) at (0,1){};
    \draw (r) to ++(1,0);
    \draw (r) to[out=90,in=0] ++(-1,.5);
    \draw (r) to[out=270,in=0] ++(-1,-.5);
    \node[dot] (s) at (0,0){};
    \draw (s) to ++(1,0);
  \end{pic}
  \enspace  \text{and} 
  \begin{pic}[scale=.6]
\node (botl) at (0, -0.5) {};
\node (topl) at (0, 0.5) {};
\draw (topl) to node[above] {} ++(.75,0) to[in=180,out=0] ++(1.5,-1) to node[below] {} ++(.75,0);
\draw (botl) to node[below] {} ++(.75,0) to[in=180,out=0]  ++(1.5,1) to node[above] {} ++(.75,0);
\end{pic}
  \]
give three distinct morphisms $2\to 2$ whereas

 \[ \begin{pic}[scale=.6]
    \node[dot] (r) at (0,0){};
    \draw (r) to ++(1,0);
    \node[dot] (s) at (0,1){};
    \draw (s) to ++(1,0);
  \end{pic}
  \enspace \text{and} \enspace 
  \begin{pic}[yscale=.6]\node[dot] (botl) at (0, -0.5) {};
\node (topl)[dot] at (0, 0.5) {};
\draw (topl) to node[above] {} ++(.6,0) to[in=180,out=0] ++(1.2,-1) to node[below] {} ++(.6,0) to[in=180,out=0]  ++(1.2,1) to node[above] {} ++(.6,0);
\draw (botl) to node[below] {} ++(.6,0) to[in=180,out=0]  ++(1.2,1) to node[above] {} ++(.6,0)  to[in=180,out=0] ++(1.25,-1) to node[below] {} ++(.6,0);
\end{pic}
  \]
denote the same morphism $0\to 2$ due to equations satisfied by any SMC (swapping twice gives the identity).

Having built an SMC from generators, we now discuss building an SMC from generators and equations.  It turns out that there is no real need to allow for equations between objects, as equations between morphisms suffice.\footnote{If one did want to identify two words $A,B\in \OO^*$, one would instead add generating morphisms $f \colon A\to B$ and $g \colon B\to A$ and equations asserting that $g$ is the inverse of $f$. While this doesn't force $A$ and $B$ to be equal in the resulting SMC, it forces them to be isomorphic, which is sufficient for most categorical purposes.} In turn, a defining equation is given by $f = g$, where $f, g $ are two morphisms $V \to W$ in the monoidal category generated by $(\OO,\MM)$.  In the resulting SMC, morphisms can now be thought of as equivalence classes of string diagrams, modulo the defining equations.

Returning to the simple example, we will now add the equations of a commutative monoid. Associativity of the multiplication is captured by 
\[
  \begin{pic}[scale=.6]
    \node[dot] (r) at (0,0){};
    \draw (r) to ++(1,0);
    \draw (r) to[out=90,in=0] ++(-1,.5) node[dot] (s) {};
    \draw (r) to[out=270,in=0] ++(-1,-.5) to ++(-1,0);
    \draw (s) to[out=90,in=0] ++(-1,.5);
    \draw (s) to[out=270,in=0] ++(-1,-.5);
  \end{pic}\enspace =\enspace
  \begin{pic}[scale=.6]
    \node[dot] (r) at (0,0){};
    \draw (r) to ++(1,0);
    \draw (r) to[out=90,in=0] ++(-1,.5)  to ++(-1,0);
    \draw (r) to[out=270,in=0] ++(-1,-.5) node[dot] (s) {};
    \draw (s) to[out=90,in=0] ++(-1,.5);
    \draw (s) to[out=270,in=0] ++(-1,-.5);
  \end{pic}
  \]
while commutativity is given by
   \[   \begin{pic}[scale=.6]
    \node[dot] (r) at (0,0){};
    \draw (r) to ++(1,0);
    \draw (r) to[out=90,in=0] ++(-1,.5) to[in=0,out=180] ++(-1.5,-1) to++(-.5,0);
    \draw (r) to[out=270,in=0] ++(-1,-.5) to[in=0,out=180] ++(-1.5,1) to++(-.5,0);
  \end{pic}\enspace =\enspace 
  \begin{pic}[scale=.6]
    \node[dot] (r) at (0,0){};
    \draw (r) to ++(1,0);
    \draw (r) to[out=90,in=0] ++(-1,.5);
    \draw (r) to[out=270,in=0] ++(-1,-.5);
  \end{pic}
  \]
and the unit laws are 
\[
  \begin{pic}[scale=.6]
    \node[dot] (r) at (0,0){};
    \draw (r) to ++(1,0);
    \draw (r) to[out=90,in=0] ++(-1,.5) node[dot] {};
    \draw (r) to[out=270,in=0] ++(-1,-.5);
  \end{pic}
  \enspace = \enspace
  \begin{pic}[baseline=-.09cm,scale=.6]
  \draw (0,0) to (2,0);
  \end{pic}
  \enspace = \enspace
    \begin{pic}[scale=.6]
    \node[dot] (r) at (0,0){};
    \draw (r) to ++(1,0);
    \draw (r) to[out=90,in=0] ++(-1,.5);
    \draw (r) to[out=270,in=0] ++(-1,-.5) node[dot] {};
  \end{pic}
\]
One can then build an SMC out of these generators and equations, whose morphisms are equivalence classes of string diagrams modulo the equations we've asserted.  In Section~\ref{ssec:baseforUC}, we will apply this machinery to UC.

%% file: 3-uc-categorical.tex
\section{UC, Categorically}
\label{sec:UCcategorical}

The theory we develop starts from an SMC $\CC$, layers some categorical constructions on top, and culminates in Theorems~\ref{thm:securemapsareanSMC} and~\ref{thm:categoricalUCcomposition} giving our composition theorems. In this section, we leave our starting SMC $\CC$ abstract. We will instantiate the general theory with a specific choice of $\CC$ in Section~\ref{ssec:baseforUC} to express UC. 

\subsection{Protocols, hybrids, and the adversary}
\label{ssec:abstractadversary}

We now assume that the starting category $\CC$ is fixed. To help intuition, one may think of the objects of $\CC$ as representing communication channels (possibly of some particular type) and of the morphisms of $\CC$ as representing networks built by connecting interactive computational systems along channels. In keeping with diagrammatic intuition, we will use the terms ``wire,'' ``port,'' and ``interface'' interchangeably to denote the endpoints of communication in our networks. Formally they correspond to identity morphisms in $\CC$. In contrast, ``tape'' is reserved for the concrete UC instantiation in Section~\ref{sec:UCconcrete}, where it refers to the usual input and subroutine-output tapes of interactive Turing machines. These also correspond to identity morphisms on the generating objects of the SMC constructed in Section~\ref{ssec:baseforUC}, thus being special types of wires.

We next use $\CC$ to construct another SMC $\CB$ and then pass to certain sub-SMCs. Ultimately, all of this is only needed to construct our main SMC of interest in Theorem~\ref{thm:securemapsareanSMC}.

As in UC, we will ensure that all protocols have explicit ``backdoor'' connections for the adversary. This is done by building an SMC $\CB$ from $\CC$. Objects of this SMC are just those of $\CC$. A morphism $X \to Y$ in $\CB$ is a  ``backdoored morphism $X\to Y$'', consisting of an object $B_f$ of $\CC$ (the type of the backdoor) and a morphism $f\colon X\to Y\otimes B_f$ in $\CC$,
under a certain equivalence relation, whose details and  motivation we will explain shortly. Intuitively, $f$ represents a network with honest interfaces given by $X$ and $Y$ and adversarial backdoors given by $B_f$. We represent such a pair pictorially by the morphism
\[
\begin{pic}
    \node[box=0/2/0/1, minimum height=8mm] (a) at (0,0) {f};
    \draw (a.west.1) to ++(-.6, 0) node[above] {$X$};
    \draw (a.east.1) to ++(.6, 0) node[right] {$Y$};
    \draw (a.east.2) to ++(.6, 0) node[right] {$B_f$};
\end{pic}
\] 
When the internal structure of a morphism in $\CB$ is not important to us, we will simply write $f \colon X\to Y$ for a morphism $[B_f,f] \colon X \to Y$. 

Pictorially, the equivalence relation is given by setting
\begin{equation}
\label{pic:equivalencerel}
\begin{pic}
    \node[box=0/2/0/1, minimum height=8mm] (a) at (0,0) {f};
    \draw (a.west.1) to ++(-.6, 0) node[above] {$X$};
    \draw (a.east.1) to ++(.6, 0) node[right] {$Y$};
    \draw (a.east.2) to ++(.6, 0) node[right] {$B_f$};
\end{pic}\sim
\begin{pic}
    \node[box=0/2/0/1, minimum height=8mm] (a) at (0,0) {f};
    \node[box=0/1/0/1, minimum height=4mm,minimum width=3mm] (b) at (1,-.2) {c};
    \draw (a.west.1) to ++(-.6, 0) node[above] {$X$};
    \draw (a.east.1) to ++(1.2, 0) node[right] {$Y$};
    \draw (a.east.2) to (b.west.1);
    \draw (b.east.1) to ++(.3, 0) node[right] {$B$};
\end{pic}
\end{equation}
whenever $c$ is an isomorphism.\footnote{We could equally well restrict to only isomorphisms generated by swaps and identities.}
Expressed symbolically, we define $(B_f,f \colon X\to Y \otimes B_f)\sim (B_g,g \colon X\to Y \otimes B_g)$ iff there is an isomorphism $c \colon B_f\to B_g$ such that $g=(\id \otimes c)\circ f$. We will use square brackets when we wish to emphasize that we are working with equivalence classes.

The equivalence relation amounts to saying that it does not matter whether an isomorphism is applied to the adversary's interface. In the case of the concrete UC category built in Section~\ref{sec:UCconcrete}, the only isomorphisms are permutations of backdoor tapes, so in this case equivalence amounts to saying that we do not care how the adversary orders its backdoors. Formally, this equivalence relation is needed so that the interchange law of monoidal categories~\eqref{eq:interchange} holds.

The composite of $[B_f,f] \colon X \to Y$ with $[B_g,g] \colon Y \to Z$ is given pictorially as 
\[
\begin{pic}
    \node[box=0/2/0/1, minimum height=8mm] (a) at (0,0) {f};
    \node[box=0/2/0/1, minimum height=4mm] (b) at (1,.205) {g};
    \draw (a.west.1) to ++(-.4, 0) node[above] {$X$} to ++(-.1,0);
    \draw (a.east.1) to (b.west.1);
    \draw (b.east.1) to ++(.5, 0) node[right] {$Z$};
    \draw (a.east.2) to ++(1.5, 0) node[right] {$B_f$};
    \draw (b.east.2) to ++(.5, 0) node[right] {$B_g$};
\end{pic}
\] 
and in symbols as  $[B_g \otimes B_f, (g \otimes \id[B_f])\circ f]$.
In words, sequential composition is performed by simply composing along the honest networks while accumulating the backdoors. 

This category inherits a symmetric monoidal structure from that of $\CC$. The monoidal product on objects is as in $\CC$, while the product of morphisms is obtained using both the monoidal product and symmetries in $\CC$: the monoidal product of $[B_1,f_1] \colon X_1 \to Y_1$ and $[B_2,f_2] \colon X_2 \to Y_2$  is given by
\[\begin{pic}
\node[box=0/2/0/1, minimum height=5mm] (a) at (0,.75) {f_1};
\draw (a.west.1) to ++(-.4, 0) node[above] {$X_1$} to ++(-.1,0);
\draw (a.east.1) to ++(1.6, 0) node[right] {$Y_1$};
\draw (a.east.2) to ++(.5, 0) to[out=0,in=180]  ++(.6,-.5) to ++(.5,0) node[right] {$B_1$};
\node[box=0/2/0/1, minimum height=5mm] (b) at (0,0) {f_2};
\draw (b.west.1) to ++(-.4, 0) node[above] {$X_2$} to ++(-.1,0);
\draw (b.east.1) to ++(.5, 0) to[out=0,in=180] ++(.6,.5) to ++(.5,0) node[right] {$Y_2$};
\draw (b.east.2) to ++(1.6, 0) node[right] {$B_2$};
\end{pic} \] 
that is, by composing in parallel and using swaps to collect the backdoors together. In symbols, the monoidal product on morphisms is defined as 
\[
    [B_1,f_1] \otimes [B_2,f_2]:=[B_1 \otimes B_2, (\id[Y_1] \otimes \sigma_{B_1,Y_2} \otimes \id[B_2]) \circ(f_1 \otimes f_2)]
\]
where $\sigma$ denotes the swap of two wires. 
We mentioned above that the equivalence relation \eqref{pic:equivalencerel} is needed so that the interchange law of monoidal categories~\eqref{eq:interchange} holds. The reader is invited to compute the two sides of the interchange law and see that they differ by a swap as the adversarial interfaces end up in different orders. This concludes our description of the SMC $\CB$. 

Our theory is defined relative to two nested sub-SMCs $\DR\subseteq  \DB$ of $\CB$. Passing from $\CB$ to a sub-SMC $\DB$ is needed to  model some idiosyncrasies of UC. In the case of UC, the main added requirement here is that every machine has exactly one backdoor channel for the adversary. In turn, the choice of a further sub-SMC $\DR$ encodes the model of corruption by specifying the allowed building blocks for real protocols. Formally, this is needed so that the task of constructing some ideal functionality does not become trivial. In the case of UC, $\DR$ will consist of networks of ITMs that are fully corruptible by the adversary, but one could choose differently: if the basic building blocks of $\DR$ are machines  which leak all information they have to the adversary without giving any control to the adversary, then one gets a  model of  honest-but-curious behavior.

\subsection{Functionalities and security}

We now assume that our SMCs $\DR\subseteq  \DB$ have been fixed. We now define what we mean by \emph{resources} in these categories.
\begin{definition}
A \emph{resource (a.k.a.\ state)} $r$ on $X$ is a morphism $r \colon I \to X$ of $\DB$.
\end{definition}

Looking ahead, in our concrete UC category (defined in Section~\ref{sec:UCconcrete}) resources will  correspond to (hybrid) protocols. In particular, any ideal functionality will be a resource, motivating the terminology.   

Formally, resources are equivalence classes of the form $[B_r,r \colon I\to X \otimes B_r]$, however, we will show that our definitions of security do not depend on the particular choice of a representative. Therefore, we will draw a resource as
\[
\begin{pic}
    \node[state,minimum width=10mm] (r) at (0,0){$r$};
    \draw (r.north east) to ++(.6,0) node[above] {$X$};
    \draw (r.south east) to ++(.6,0) node[below] {$B_r$};
  \end{pic}
\]

Before giving our security definition, we need one further ingredient, namely, an abstract notion of indistinguishability. In the case of UC, this will coincide with the usual notion of  computational indistinguishability (defined formally in section~\ref{ssec:indistinguishability}). For now, all we need to know about computational indistinguishability is that it gives, for each object $X$ of $\CC$, an equivalence relation $\approx$ on states $I\to X$ in $\CC$, and that these equivalence relations behave well with respect to sequential and parallel composition in the sense defined below.

\begin{definition}\label{def:indistinguishability}
Given an equivalence relation $\approx_X$ on states $I\to X$ in $\CC$ for every object $X$,  we say that $\approx:=(\approx_X)_{X\in\CC}$ \emph{respects (sequential and parallel) composition} if 
\begin{enumerate}
    \item If $r\approx_X s$ for $r,s \colon I\to X$ in $\CC$ and  $f \colon X\to Y$ is a morphism in $\CC$, then $f\circ r\approx_Y f\circ  s$. 
    \item If $r_i,s_i \colon I\to X_i$ for $i=1,2$ are states in $\CC$ such that $r_i\approx_{X_i} s_i$, then $r_1 \otimes r_2\approx_{X_1 \otimes X_2} s_1 \otimes s_2$.
\end{enumerate}
\end{definition}
From now on, we omit the subscripts from $\approx$ as they can be inferred from context.

We now give our security definition, which follows the simulation paradigm by quantifying over all possible actions of the adversary.
\begin{definition}\label{def:security}
    Let $r \colon I\to X$ and $s \colon I \to Y$ be resources. A morphism $f \colon X\to Y$ of $\DR$ 
    is a \emph{secure emulation}  $r\to s$ if for every object $B$ and every morphism  $a \colon B_f \otimes B_r\to B$ (corresponding to the adversary) in $\CC$ 
there exists a $b \colon B_s\to B$ in $\CC$ (corresponding to a simulator) such that
   \begin{equation}\label{eq:UCsecurity}
\begin{pic}
    \node[state,minimum width=1cm] (r) at (0,0){$r$};
    \node[box=0/2/0/1, minimum height=5.3mm] (a) at (1,.23) {f};
    \node[box=0/1/0/2, minimum height=5.3mm] (b) at (2,-.08) {a};
    \draw (r.north east) to (a.west.1);
    \draw (r.south east) to ([yshift=-.5pt]b.west.2);
    \draw ([yshift=.5pt]a.east.1) to ++(1.6, 0) node[above] {$Y$};
    \draw ([yshift=-0.5pt]a.east.2) to ([yshift=0.5pt]b.west.1);
    \draw (b.east.1) to ++(.6, 0) node[below] {$B$};
\end{pic}\, \approx \,
\begin{pic}
    \node[state,minimum width=1cm] (r) at (0,0){$s$};
    \node[box=0/1/0/1, minimum height=5mm] (b) at (1.1,-.23) {b};
    \draw (r.north east) to ++(1.7,0) node[above] {$Y$};
    \draw (r.south east) to (b.west.1);
    \draw (b.east.1) to ++(.5, 0) node[below] {$B$};
  \end{pic}\end{equation}
 \ie 
  $(\id[Y] \otimes a) \circ (f \otimes \id[B_r]) \circ r
  \;\approx\;
  (\id[Y] \otimes b) \circ s $ in $\CC$.
\end{definition}

Note that this is more or less the picture people would usually draw when defining simulation-based security, except that (i)~it is now a \emph{rigorous} string diagram and (ii)~the equivalence relation $\approx$ implicitly quantifies over all environments, and as a result the environment is not drawn (see Lemma~\ref{lem:indistinguishability} for a formal justification of this in the case of UC). Our completeness of the dummy adversary Theorem~\ref{thm:categoricalcompletenessofdummy} will allow us to also omit the universally quantified $a$ from the picture.

A priori, the definition of security depends on the choice of a representative. We now observe that this is not the case. The idea in the proof is simple: equivalence classes only differ by applying an isomorphism on the adversarial port, and using this we can show that one choice of representatives satisfies security if any other does. 

\begin{proposition}
    Security is well-defined.
\end{proposition}

\begin{proof}
    Assume $r$, $s$ and $f$ satisfy Definition~\ref{def:security}. Consider arbitrary $r',s',f'$ with $r\sim r'$, $s\sim s'$ and $f\sim f'$ and fix isomorphisms $c_r \colon B_r\to B_{r'}$, $c_s \colon B_s\to B_{s'}$ and $c_f \colon B_f\to B_{f'}$ witnessing this. To show that $r',s'$ and $f'$ also satisfy  Definition~\ref{def:security}, consider an arbitrary adversary $a \colon B_{f'} \otimes B_{r'}\to B$ in $\CC$, and compute as follows: 
\begin{align*}\begin{pic}
    \node[state,minimum width=1cm] (r) at (0,0){$r'$};
    \node[box=0/2/0/1, minimum height=5.3mm] (a) at (1,.23) {f'};
    \node[box=0/1/0/2, minimum height=5.3mm] (b) at (2,-.08) {a};
    \draw (r.north east) to (a.west.1);
    \draw (r.south east) to ([yshift=-.5pt]b.west.2);
    \draw ([yshift=.5pt]a.east.1) to ++(1.6, 0) node[above] {};
    \draw ([yshift=-0.5pt]a.east.2) to ([yshift=0.5pt]b.west.1);
    \draw (b.east.1) to ++(.6, 0) node[above] {};
\end{pic} \, &= \,
\begin{pic}
    \node[state,minimum width=2cm] (r) at (0,0) {$r$};
    \node[box=0/2/0/1, minimum height=7mm] (f) at (1,.43) {f};
    \node[box=0/1/0/2, minimum height=12mm] (a) at (3,-.1) {a};
    \node[box=0/1/0/1, minimum height=4mm] (cf) at (2,.225) {c_f};
    \node[box=0/1/0/1, minimum height=4mm] (cr) at (2,-.43) {c_r};
    \draw (r.north east) to (f.west.1);
    \draw (r.south east) to (cr.west);
    \draw (cr.east) to([yshift=-.75pt]a.west.2);
    \draw ([yshift=.5pt]f.east.1) to ++(2.6, 0) node[above] {};
    \draw ([yshift=-0.75pt]f.east.2) to (cf.west);
    \draw (cf.east) to ([yshift=.5pt]a.west.1);
    \draw (a.east.1) to ++(.6, 0) node[above] {};
        \node[fit=(a) (cf) (cr), draw, dashed,yshift=-1.5pt] (bigbox) {};
\end{pic} \\
\approx 
\begin{pic}
    \node[state,minimum width=1cm] (r) at (0,0){$s$};
    \node[box=0/1/0/1, minimum height=5mm] (b) at (1.1,-.23) {b};
    \draw (r.north east) to ++(1.7,0) node[above] {};
    \draw (r.south east) to (b.west.1);
    \draw (b.east.1) to ++(.5, 0) node[above] {};
  \end{pic} &=
  \begin{pic}
    \node[state,minimum width=1cm] (r) at (0,0){$s'$};
    \node[box=0/1/0/1, minimum height=5mm] (cs) at (1,-.23) {c_s^{-1}};
    \node[box=0/1/0/1, minimum height=6mm] (b) at (2,-.23) {b};
    \draw (r.north east) to ++(2.6,0) node[above] {};
    \draw (r.south east) to (cs.west);
    \draw (cs.east) to (b.west);
    \draw (b.east) to ++(.5, 0) node[above] {};
        \node[fit=(b) (cs), draw, dashed,yshift=-1pt] (bigbox) {};
  \end{pic}\end{align*}
where the first equation follows from the defining properties of $c_f$ and $c_r$, the second step follows by security of $f \colon r\to s$, and the final equation follows from properties of $c_s$.
\end{proof}

To check security, there is a convenient special case which is sufficient. 
\begin{theorem}
[Completeness of the dummy adversary] 
\label{thm:categoricalcompletenessofdummy}  Let $r \colon I\to X$ and $s \colon I \to Y$ be resources. If a morphism $f \colon X\to Y$ in $\DR$ satisfies
   \begin{equation}\label{eq:UCsecuritywithdummy}
\begin{pic}
    \node[state,minimum width=1cm] (r) at (0,0){$r$};
    \node[box=0/2/0/1, minimum height=5mm] (a) at (1,.23) {f};
    \draw (r.south east) to ++(1.65,0) node[right] {$B_r$};
    \draw (r.north east) to (a.west.1);
    \draw ([yshift=1pt]a.east.1) to ++(.5, 0) node[right] {$Y$};
    \draw ([yshift=-1pt]a.east.2) to ++(.5,0) node[right] {$B_f$};
  \end{pic}  \enspace\approx \enspace
\begin{pic}
    \node[state,minimum width=1cm] (r) at (0,0){$s$};
    \node[box=0/2/0/1, minimum height=5mm] (b) at (1,-.23) {b};
    \draw (r.north east) to ++(1.65,0) node[right] {$Y$};
    \draw (r.south east) to (b.west.1);
    \draw (b.east.1) to ++(.5, 0) node[right] {$B_f$};
    \draw (b.east.2) to ++(.5, 0) node[right] {$B_r$};
  \end{pic}\end{equation}
\ie $  (f \otimes \id[B_r]) \circ r
  \;\approx\;
  (\id[Y] \otimes b) \circ s$
 for some $b$,  then $f$ is a secure emulation $r\to s$.
\end{theorem}
\begin{proof}
Let us assume \eqref{eq:UCsecuritywithdummy}. Then, as $\approx$ respects composition, 
 \[ \begin{pic}
    \node[state,minimum width=1cm] (r) at (0,0){$r$};
    \node[box=0/2/0/1, minimum height=5.3mm] (a) at (1,.23) {f};
    \node[box=0/1/0/2, minimum height=5.3mm] (b) at (2,-.08) {a};
    \draw (r.north east) to (a.west.1);
    \draw (r.south east) to ([yshift=-.5pt]b.west.2);
    \draw ([yshift=.5pt]a.east.1) to ++(1.6, 0) node[above] {};
    \draw ([yshift=-0.5pt]a.east.2) to ([yshift=0.5pt]b.west.1);
    \draw (b.east.1) to ++(.6, 0) node[above] {}; 
\end{pic} \approx 
  \begin{pic}
    \node[state,minimum width=1cm] (r) at (0,0){$s$};
    \node[box=0/2/0/1, minimum height=5mm] (b) at (1,-.23) {b};
    \node[box=0/1/0/2, minimum height=5mm] (a) at (2,-.23) {a};
    \draw (r.north east) to ++(2.6,0) node[above] {};
    \draw (r.south east) to (b.west.1);
    \draw (b.east.1) to  (a.west.1);
    \draw (b.east.2) to  (a.west.2);
    \draw (a.east.1) to ++(.5, 0) node[above] {};
  \end{pic}\]
holds for all attackers  $a\colon B_f\otimes B_r\to B$ in $\CC$, giving us~\eqref{eq:UCsecurity}. 
Symbolically, this derivation becomes 
\begin{align*}
(\id[Y] \otimes a) \circ (f \otimes \id[B_r]) \circ r
&\approx (\id[Y] \otimes a) \circ (\id[Y] \otimes b) \circ s \\
&= (\id[Y] \otimes (a \circ b)) \circ s. 
\quad \qedhere
\end{align*}

\end{proof}
This result is an analogue of the ``completeness of the dummy adversary'' in UC. However, the identity maps are not literally the dummy adversary in UC, as the latter has to multiplex several channels into the only channel it exposes to the environment. We prove a result capturing completeness for a wide class of adversaries, which will be needed later for our translation Theorem~\ref{thm:translation} into UC. 

\begin{corollary}
\label{cor:splitmonocompleteness} 
 Let $r \colon I\to X$ and $s \colon I \to Y$ be resources. Assume $f\colon X\to Y$ in $\DR$ satisfies the security equation~\eqref{eq:UCsecurity} for some fixed $a$ and $b$ in $\CC$. Assume moreover that $a$ is
\emph{split mono on $f$ applied to $r$} in the sense that
there is a morphism $\bar a$ in $\CC$ such that 
\begin{equation}\label{eq:splitmono}
\begin{pic}
    \node[state,minimum width=1cm] (r) at (0,0){$r$};
    \node[box=0/2/0/1, minimum height=5.3mm] (f) at (1,.23) {f};
    \node[box=0/1/0/2, minimum height=5.3mm] (a) at (2,-.08) {a};
    \node[box=0/2/0/1, minimum height=5.3mm] (ai) at (3,-.08) {\bar a};
    \draw (r.north east) to (f.west.1);
    \draw (r.south east) to ([yshift=-.5pt]a.west.2);
    \draw ([yshift=.5pt]f.east.1) to ++(2.6, 0) node[above] {};
    \draw ([yshift=-0.5pt]f.east.2) to ([yshift=0.5pt]a.west.1);
    \draw (a.east.1) to (ai.west.1);
    \draw ([yshift=.5pt]ai.east.1) to ++(.6, 0) node[above] {};
    \draw ([yshift=-.5pt]ai.east.2) to ++(.6, 0) node[below] {};
\end{pic}
\enspace\approx\enspace
\begin{pic}
    \node[state,minimum width=1cm] (r) at (0,0) {$r$};
    \node[box=0/2/0/1, minimum height=5mm] (f) at (1,.23) {f};
    \draw (r.south east) to ++(1.65,0) node[below] {};
    \draw (r.north east) to (f.west.1);
    \draw ([yshift=1pt]f.east.1) to ++(.5, 0) node[above] {};
    \draw ([yshift=-1pt]f.east.2) to ++(.5,0);
\end{pic}
\end{equation}
\ie   $(\id[Y] \otimes (\bar a \circ a)) \circ (f \otimes \id[B_r]) \circ r
  \;\approx\;
  (f \otimes \id[B_r]) \circ r$.
Then $f$ is a secure emulation $r\to s$.
\end{corollary}
\begin{proof} Attach $\bar{a}$ to bottom wires on both sides of~\eqref{eq:UCsecurity}: then equation~\eqref{eq:splitmono} implies~\eqref{eq:UCsecuritywithdummy} which is sufficient by Theorem~\ref{thm:categoricalcompletenessofdummy}.
\end{proof}
In the context of UC, the morphism $a$ above will be the actual UC dummy adversary.

\subsection{Categorical composition theorems}

For us, the core of the composition theorems is that \emph{secure emulation maps form an SMC} as this captures  that secure maps are closed under sequential and parallel composition, \emph{and} also that these compositions are suitably well-behaved. In our opinion, this is what Abstract Cryptography gestures at when talking about  ``composition-order invariance''~\cite[Definition~13]{MR11}. We now give our composition theorem.

\begin{theorem}[SMC structure on secure emulations]
\label{thm:securemapsareanSMC}
    Resources and secure emulations between them form a symmetric monoidal category.
\end{theorem}

\begin{proof}
We repeatedly use Theorem~\ref{thm:categoricalcompletenessofdummy}.
We first observe that secure emulations are closed under sequential composition. Assuming $f \colon r\to s$ and $g \colon s\to t$ are secure, so is $g \circ f \colon r\to t$:
\begin{align*}\begin{pic}
\node[state,minimum width=1cm] (r) at (0,0){$r$};
\node[box=0/2/0/1, minimum height=6mm] (a) at (.9,.23) {f};
\node[box=0/2/0/1, minimum height=4mm] (b) at (1.8,.385) {g};
\draw (r.north east) to (a.west.1);
\draw (r.south east) to ++(2.3,0);
\draw (a.east.1) to (b.west.1);
\draw (b.east.1) to ++(.35, 0) node[above] {};
\draw (a.east.2) to ++(1.25, 0) node[above] {};
\draw (b.east.2) to ++(.35, 0) node[below] {};
\end{pic}
 \ &\approx  \ 
 \begin{pic}
    \node[state,minimum width=1cm] (r) at (0,0){$s$};
    \node[box=0/2/0/1, minimum height=5mm] (b) at (.9,-.23) {b};
    \node[box=0/2/0/1, minimum height=5mm] (a) at (1.8,.23) {g};
    \draw (r.north east) to (a.west.1);
    \draw (r.south east) to (b.west.1);
    \draw (b.east.1) to ++(1.25, 0) node[above] {};
    \draw (b.east.2) to ++(1.25, 0) node[above] {};
    \draw ([yshift=1pt]a.east.1) to ++(.35, 0) node[above] {};
    \draw ([yshift=-1pt]a.east.2) to ++(.35,0);
  \end{pic} = \\ 
  \begin{pic}
    \node[state,minimum width=1cm] (s) at (0,0){$s$};
    \node[box=0/2/0/1, minimum height=3mm] (b) at (1.8,-.23) {b};
    \node[box=0/2/0/1, minimum height=5mm] (g) at (.9,.23) {g};
    \draw (s.north east) to (g.west.1);
    \draw (s.south east) to (b.west.1);
    \draw (b.east.1) to ++(.35, 0);
    \draw (b.east.2) to ++(.35, 0);
    \draw ([yshift=.5pt]g.east.1) to ++(1.25, 0);
    \draw ([yshift=-.5pt]g.east.2) to ++(1.25,0);
  \end{pic} \ &\approx \ 
  \begin{pic}
\node[state,minimum width=1cm] (t) at (0,0){$t$};
\node[box=0/2/0/1, minimum height=6mm] (a) at (.9,-.23) {a};
\node[box=0/2/0/1, minimum height=4mm] (b) at (1.8,-.385) {b};
\draw (t.south east) to (a.west.1);
\draw (t.north east) to ++(2.3,0);
\draw (a.east.2) to (b.west.1);
\draw (b.east.1) to ++(.35, 0) node[above] {};
\draw (b.east.2) to ++(.35, 0) node[above] {};
\draw (a.east.1) to ++(1.25, 0) node[below] {};
\draw (b.east.2) to ++(.35, 0) node[below] {};
\end{pic}\end{align*}
where the first (third) step follows from security of $f$ (of $g$) and $\approx$ respecting composition, and the second equation follows from the interchange law. 
Equationally,
\begin{align*}
&(g \otimes \id[B_f \otimes B_r]) \circ (f \otimes \id[B_r]) \circ r \\
&\approx (g \otimes \id[B_f \otimes B_r]) \circ (\id[Y] \otimes b) \circ s
  && \text{(security of } f\text{)} \\
&= (g \otimes b) \circ s
  && \text{(interchange)} \\
&= (\id[Z \otimes B_g] \otimes b) \circ (g \otimes \id[B_s]) \circ s
  && \text{(interchange)} \\
&\approx (\id[Z \otimes B_g] \otimes b) \circ (\id[Z] \otimes a) \circ t
  && \text{(security of } g\text{)} \\
&= \bigl(\id[Z] \otimes ((\id[B_g] \otimes b) \circ a)\bigr) \circ t.
  && \text{(interchange)}
\end{align*}

To show that secure maps are closed under parallel composition, we let $f \colon r\to s$ and $g \colon t\to u$ be secure and compute as follows: 
\[
\begin{pic} 
    \node[state,minimum width=1cm] (r) at (0,0){$r$};
    \node[box=0/2/0/1, minimum height=5mm] (f) at (1,.23) {f};
    \draw (r.south east) to[in=180,out=0] ++(.5,-.55) to ++(1.15,0) node[below] {};
    \draw (r.north east) to (f.west.1);
    \draw ([yshift=1pt]f.east.1) to ++(.5, 0) node[above] {};
    \draw ([yshift=-1pt]f.east.2) to[in=180,out=0]   ++(.5,-.25);
    \node[state,minimum width=1cm] (t) at (0,-1.1){$t$};
    \node[box=0/2/0/1, minimum height=5mm] (b) at (1,-.35) {g};
    \draw (t.south east) to ++(1.65,0) node[below] {};
    \draw (t.north east) to[in=180,out=0] (b.west.1);
    \draw ([yshift=1pt]b.east.1) to[in=180,out=0] ++(.5,.25) node[above] {};
    \draw ([yshift=-1pt]b.east.2) to ++(.5,0);
  \end{pic} \ = \  
\begin{pic}
    \node[state,minimum width=1cm] (r) at (0,0){$r$};
    \node[box=0/2/0/1, minimum height=5mm] (f) at (1,.23) {f};
    \draw (r.south east) to ++(1.13,0) to[in=180,out=0] ++(.52,-.47) to[in=180,out=0] ++(.3,-.33);
    \draw (r.north east) to (f.west.1);
    \draw ([yshift=1pt]f.east.1) to ++(.8, 0);
    \draw ([yshift=-1pt]f.east.2) to ++(.5,0) to[in=180,out=0] ++(.3,-.295);
    \node[state,minimum width=1cm] (t) at (0,-1.1){$t$};
    \node[box=0/2/0/1, minimum height=5mm] (g) at (1,-0.87) {g};
    \draw (t.south east) to ++(1.95,0);
    \draw (t.north east) to (g.west.1);
    \draw ([yshift=1pt]g.east.1) to[in=180,out=0] ++(.5, .475)  to[in=180,out=0] ++(.3,.295);
    \draw ([yshift=-1pt]g.east.2) to ++(0.5,0) to[in=180,out=0] ++(.3,.33);
  \end{pic}  \ \approx \ 
\begin{pic}
  \node[state,minimum width=1cm] (s) at (0,0){$s$};
  \node[box=0/2/0/1, minimum height=5mm] (a) at (1,-.23) {a};
  \draw (s.north east) to ++(1.95,0);
  \draw (s.south east) to (a.west.1);
  \draw ([yshift= 1pt]a.east.1)
    to[out=-10, in=190] ++(.8,-.8);  
  \draw ([yshift=-1pt]a.east.2)
    to[out=-20, in=200] ++(.8,-.8);
  \node[state,minimum width=1cm] (u) at (0,-1.1){$u$};
  \node[box=0/2/0/1, minimum height=5mm] (b) at (1,-1.33) {b};
  \draw (u.north east)
    to[out=0, in=180] ++(1.12,0)  
    to[out=0, in=170] ++(.8,.8); 
  \draw (u.south east) to (b.west.1);
  \draw ([yshift= 1pt]b.east.1)
    to[out=10, in=-170] ++(.8,.8);  
  \draw ([yshift=-1pt]b.east.2)
    to[out=0, in=180] ++(.8,0);    
\end{pic}
\]
The first equation is true in any monoidal category, and the second equation follows from security of $f$ and $g$ and from $\approx$ respecting composition.  

It is easy to show that swaps, and in fact all isomorphisms, are secure. Moreover, as our SMC structure is inherited from the SMC $\CB$, the further required properties of an SMC  automatically hold. 
\end{proof}

We now move on to a categorical counterpart of the composition theorem of UC (for static systems)~\cite[Theorem 3]{Can20}. 
\begin{theorem}[Categorical universal composition]
\label{thm:categoricalUCcomposition}
Let $r,s \colon I\to X$ be resources. Assume that the identity on $X$ gives a secure emulation $r\to s$. Then any $f \colon X\to Y$ in $\DR$ 
 gives a secure emulation from $r$ to $f\circ s$.
\end{theorem}
\begin{proof}
By assumption there exists a simulator $b$ such that 
   \[ \begin{pic}
    \node[state,minimum width=1cm] (r) at (0,0){$r$};
    \draw (r.north east) to ++(.6,0) node[above] {$X$};
    \draw (r.south east) to ++(.6,0) node[below] {$B_r$};
  \end{pic} \enspace \approx  \enspace \begin{pic}
    \node[state,minimum width=1cm] (r) at (0,0){$s$};
    \node[box=0/1/0/1, minimum height=5mm,minimum width=4mm] (b) at (1.1,-.23) {b};
    \draw (r.north east) to ++(1.7,0) node[above] {$X$};
    \draw (r.south east) to (b.west.1);
    \draw (b.east.1) to ++(.5, 0) node[below] {$B_r$}; 
  \end{pic}.\]
It then follows that 
\[ \begin{pic}
    \node[state,minimum width=1cm] (r) at (0,0){$r$};
    \node[box=0/2/0/1, minimum height=5mm,minimum width=4mm] (a) at (.7,.23) {f};
    \draw (r.south east) to ++(1.1,0) node[right] {$B_r$};
    \draw (r.north east) to (a.west.1);
    \draw ([yshift=1pt]a.east.1) to ++(.35, 0) node[right] {$Y$};
    \draw ([yshift=-1pt]a.east.2) to ++(.35,0) node[right] {$B_f$};
  \end{pic} \approx  \ \begin{pic}
    \node[state,minimum width=1cm] (r) at (0,0){$s$};
    \node[box=0/1/0/1, minimum height=5mm,minimum width=4mm] (b) at (.7,-.23) {b};
    \node[box=0/2/0/1, minimum height=5mm,minimum width=4mm] (a) at (1.4,.23) {f};
    \draw (r.north east) to (a.west.1);
    \draw (r.south east) to (b.west.1);
    \draw (b.east.1) to ++(1.05, 0) node[right] {$B_r$};
    \draw ([yshift=1pt]a.east.1) to ++(.35, 0) node[right] {$Y$};
    \draw ([yshift=-1pt]a.east.2) to ++(.35,0) node[right] {$B_f$};
  \end{pic} =  \begin{pic}
    \node[state,minimum width=1cm] (r) at (0,0){$s$};
    \node[box=0/1/0/1, minimum height=3mm,minimum width=4mm] (b) at (1.4,-.23) {b};
    \node[box=0/2/0/1, minimum height=5mm,minimum width=4mm] (a) at (.7,.23) {f};
    \draw (r.north east) to (a.west.1);
    \draw (r.south east) to (b.west.1);
    \draw (b.east.1) to ++(.35, 0) node[right] {$B_r$};
    \draw ([yshift=.5pt]a.east.1) to ++(1.05, 0) node[right] {$Y$};
    \draw ([yshift=-.5pt]a.east.2) to ++(1.05,0) node[right] {$B_f$};
  \end{pic}\]
where the first equation follows from the previous equation and indistinguishability respecting composition, and the second from the axioms of an SMC. 
We have now proved the result for the dummy adversary, which is sufficient by Theorem~\ref{thm:categoricalcompletenessofdummy}. 
\end{proof} 

\subsection{Global subroutines}

We now observe that working with shared ``global'' resources (\eg a shared clock or a shared PKI) already falls under the theory above. Indeed, a ``resource'' $r$ potentially calling (and therefore connected to) a global resource $\gamma \colon I\to Y \otimes X$ would be modelled as a morphism $r \colon X\to W$, and the resulting total resource $(\id \otimes r)\circ \gamma$ would be depicted as 
\[\begin{pic}
    \node[state,minimum width=1.25cm] (r) at (0,0){$\gamma$};
    \node[box=0/2/0/1, minimum height=5mm] (a) at (1,0) {r};
    \draw ([yshift=-5pt]r.south east) to ++(1.65,0) node[below] {};
    \draw ([yshift=5pt]r.north east) to ++(1.65,0) node[below] {};
    \draw (r.east) to (a.west.1);
    \draw ([yshift=1pt]a.east.1) to ++(.5, 0) node[above] {};
    \draw ([yshift=-1pt]a.east.2) to ++(.5,0);
  \end{pic} \]
where the top wire represents the (in general composite) remaining interface $Y$ of $\gamma$, the one below it is the honest interface of $r$, and the bottom two wires give the interfaces of the adversary. We can then define secure emulations $r\to s$ in the presence of a global  resource $\gamma$ as  ordinary secure emulations  from $r$ with $\gamma$ to $s$ with $\gamma$.
\begin{definition}
    A secure emulation $r\to s$ \emph{with a global resource $\gamma$} is a secure emulation $(\id \otimes r)\circ\gamma \to(\id \otimes s )\circ\gamma$.
\end{definition}
We now observe that the Universal Composition with Global Subroutines (UCGS) Theorem~\cite{CDP+:globalUC} for static systems can be captured as a special case of Theorem~\ref{thm:categoricalUCcomposition}, following~\cite{BCH+:globalwithUC}.
\begin{theorem}[Categorical universal composition with global subroutines]
\label{thm:categoricalUCcompositionglobal} 
    Assume that the identity map gives a secure emulation from $r \colon X\to W$ with global resource $\gamma$ to $s \colon X\to W$ with $\gamma$. Then any $f \colon Y \otimes W\to W'$ in $\DR$ gives a secure emulation from $(\id \otimes r )\circ\gamma$ (\ie from $r$ with $\gamma$) to 
\[\begin{pic}
    \node[state,minimum width=1.5cm] (g) at (0,0){$\gamma$};
    \node[box=0/2/0/1, minimum height=5mm] (s) at (1,0) {s};
    \node[box=0/2/0/2, minimum height=7mm] (f) at (2,.33) {f};
    \draw ([yshift=-5pt]g.south east) to ++(2.5,0) node[below] {};
    \draw ([yshift=5pt]g.north east) to (f.west.1);
    \draw (g.east) to (s.west.1);
    \draw ([yshift=.6pt]s.east.1) to (f.west.2);
    \draw ([yshift=-.6pt]s.east.2) to ++(1.5,0);
    \draw (f.east.1) to ++(.5, 0) node[above] {};
    \draw (f.east.2) to ++(.5,0);
  \end{pic} \]
 \ie to the resource $f\circ (\id \otimes s)\circ\gamma$. 
\end{theorem}

In Section~\ref{ssec:translatingtoUC} we discuss how the above theorem can be used to derive the UCGS theorem of~\cite{BCH+:globalwithUC} as a corollary.

%% file: 4-uc-concrete.tex
\section{UC, Concretely}
\label{sec:UCconcrete}

One could instantiate the theory developed in the previous section with any of the example SMCs discussed so far. However, these are unlikely to yield cryptographically meaningful theories. Instead, one can extract SMCs fairly directly from (minor variants of) state-separating proofs~\cite{BDF+18}, random systems~\cite{M02}, and perhaps with more effort from elsewhere in cryptography. In general, finding cryptographically meaningful categories off-the-shelf is difficult (cf.~\cite[Section~9]{BK23}). 

As a result, one might want to build this SMC directly. This is particularly likely to be useful if one has some notion of an atomic ``interactive computational system,'' which when composed results in \emph{a network} of such systems (\eg two ITMs can be connected to form a network of two ITMs, rather than a single ITM)\footnote{Whether or not the composed network can be simulated by a single ITM is beside the point here.} instead of, say, how two functions compose into a single function. 

We can conveniently build such categories by means of generators and equations as described in Section~\ref{ssec:gensandrels}. When describing the resulting category, it is useful to borrow intuitions from UC. In this viewpoint, the generating objects  $\OO$ correspond to basic types of communication channels. As a first approximation, UC has just one basic type, but one could imagine communication channels parametrized by the type of data flowing on them (so that a channel for bit strings is different from a channel for communicating vectors, graphs, or elements of a group), or perhaps each channel could have a session type~\cite{bartoletti_combining_2015}, capturing the expected communication pattern (``first I send you a vertex of a graph, then you respond with an integer, giving its color'' etc.). Given two words $V,W$, one then thinks of the generating set $\MM_{V,W}$ as the set of all atomic/basic computational systems with channels $V$ on one side of the system and channels $W$ on the other. For example, in our formalization of UC these will correspond to interactive Turing machines with a given number of input tapes and a given number of subroutine-output tapes. One may also add further generators and equations if need be. We note that the backdoor tape is not modeled in this base category and will be added later on. 

\subsection{The starting category for UC}
\label{ssec:baseforUC}

We will now describe how to build, via generators and relations, an SMC whose morphisms are networks of interactive Turing machines (ITMs). We will fix details of our machines as needed; for now, it is sufficient to assume that in addition to some internal tapes each given ITM has some number $n$ of tapes for receiving inputs from other ITMs (from now on, input tapes), and some number $m$ of tapes for receiving subroutine-outputs from other ITMs (from now on, subroutine-output tapes or just output tapes). We assume that these tapes are enumerated, so that one can always speak without ambiguity of the $i$th input tape (or the $j$th subroutine-output tape). Connections between ITMs are then obtained by pairing the $i$th input tape of one machine with the $j$th subroutine-output tape of another, letting each machine send messages that are received by the other machine at the corresponding tape. (No tape gets paired twice.) Thus, as in UC, the tapes are only used for receiving messages. However, we differ from UC in that we allow multiple subroutine-output and input tapes, instead of one each. We will return to this in Section~\ref{ssec:translatingtoUC}.

Formally, this pairing of tapes takes place in the SMC that we will build, with its operational interpretation given in Section~\ref{ssec:indistinguishability}. However, the picture to keep in mind is that when such a pairing is done, one machine can receive inputs from the other machine on its $i$th input tape, and can send outputs to the other machine's $j$th subroutine-output tape.

As a first approximation, UC has only one type of wire, representing an untyped communication channel. However, the default kind of communication between ITMs in UC is never between equals: given two communicating machines, one is always designated to be the ``caller'' and the other the ``subroutine.'' While formally UC does not enforce a difference between these two roles, they are often used in practice. 
We will thus set $\OO:=\{A_\rightarrow,A_\leftarrow\}$, with $A_\rightarrow$ representing a connection with the caller on the left and $A_\leftarrow$ representing a connection with the caller on the right. Instead of labeling the diagrams with $A_\rightarrow$ and $A_\leftarrow$ for these two objects, we will simply draw corresponding identity morphisms by an arrow pointing away from the caller: 
\[
\begin{pic}
    \draw (0,0) to ++(.5,0) node {$>$} to ++(.5,0);
\end{pic} 
\enspace \text{for $\id[A_\rightarrow]$ \quad and \quad } 
\enspace 
\begin{pic}
    \draw (0,0) to ++(.5,0) node {$<$} to ++(.5,0);
\end{pic} \enspace \text{for $\id[A_\leftarrow]$}
\] 
We define $\MM_{A_\leftarrow^m,A_\leftarrow^n}$ to consist of the set of ITMs with $m$ subroutine-output tapes and $n$ input tapes. We will explain shortly why we do not include words containing $A_\rightarrow$.

If these were all the generators we had, we could only build DAGs of ITMs whereas UC imposes no restrictions on the shape of connectivity. For example, one could have a cycle of calls, consisting of three machines $f,g,h$, with $h$ calling $g$, $g$ calling $f$, and $f$ calling $h$, resulting in the diagram 
\[\begin{pic}
\node[box=0/1/0/1] (f) at (0,0) {f};
\node[box=0/1/0/1] (g) at (1.5,0) {g};
\node[box=0/1/0/1] (h) at (3,0) {h};
\draw (f.east) to  node[midway] {$<$}  (g.west);
\draw (g.east) to  node[midway] {$<$}  (h.west);
\draw (h.east) to ++(.3, 0) to[in=0,out=0] ++(0,.6) to  node[midway] {$>$}  ++(-4.32,0) to[in=180,out=180] ++(0,-.6) to (f.west);
\end{pic}\] 

To make sense of this, we need the ability to achieve such connectivity.  Formally, we add two additional generating morphisms 
by setting $\MM_{I,A_\leftarrow A_\rightarrow}=\{\eta\}$ and $\MM_{A_\rightarrow A_\leftarrow, I}=\{\epsilon\}$
and the defining equations 
\begin{align*}
  (\epsilon \otimes \id[A_{\rightarrow}]) \circ (\id[A_{\rightarrow}] \otimes \eta)
  &= \id[A_{\rightarrow}]
 \\
  (\id[A_{\leftarrow}] \otimes \epsilon) \circ (\eta \otimes \id[A_{\leftarrow}])
  &=\id[A_{\leftarrow}]~.
\end{align*}
We first indulge in some syntactic sugar by denoting
\[ 
      \begin{pic}
    \node[state,minimum width=10mm] (r) at (0,0){$\eta$};
    \draw (r.north east) to ++(.3,0) node {$<$} to ++(.3,0);
    \draw (r.south east) to ++(.3,0) node {$>$} to ++(.3,0);
  \end{pic} \enspace \text{as } \enspace 
  \begin{pic}[scale=.6]
    \draw (0,0) to++(-.2,0) node {$>$} to[out=180,in=-90]  ++(-.5,.5) to[out=90,in=180] ++(.5,.5) node {$<$} to ++(.2,0);
  \end{pic}\enspace \quad\text{and}\quad \enspace 
      \begin{pic}
    \node[costate,minimum width=10mm] (r) at (0,0){$\epsilon$};
    \draw (r.north west) to ++(-.3,0) node {$>$} to ++(-.3,0);
    \draw (r.south west) to ++(-.3,0) node {$<$} to ++(-.3,0);
  \end{pic}  \enspace \text{as }
  \begin{pic}[scale=.6]
    \draw (0,0) to ++(.2,0) node {$<$}  to[out=0,in=-90] ++(.5,.5) to[out=90,in=0] ++(-.5,.5) node {$>$}  to ++(-.2,0);
  \end{pic} 
  \]
The two defining equations above then correspond to the following \emph{snake equations}:
\[ 
   \begin{pic}[scale=.6]
    \draw (0,0) to ++(-.5,0) node {$>$} to++(-.5,0) to[out=180,in=-90] ++(-.5,.5)   to[out=90,in=180] ++(.5,.5) node {$<$} to[out=0,in=-90] ++(.5,.5) to[out=90,in=0] ++(-.5,.5)  to ++(-.5,0) node {$>$} to++(-.5,0);
  \end{pic}  \enspace = \enspace
  \begin{pic}[baseline=-.09cm,scale=.6]
  \draw (0,0)  to (1,0) node {$>$}  to (2,0);
  \end{pic} \enspace \text{and} \enspace 
  \begin{pic}[scale=.6]
    \draw (0,0)  to ++(.5,0) node {$<$} to ++(.5,0) to[out=0,in=-90] ++(.5,.5) to[out=90,in=0] ++(-.5,.5)  node {$>$} to[out=180,in=-90] ++(-.5,.5) to[out=90,in=180] ++(.5,.5)  to ++(.5,0) node {$<$} to ++(.5,0);
  \end{pic}  \enspace = \enspace
  \begin{pic}[baseline=-.09cm,scale=.6]
  \draw (0,0) to (1,0) node {$<$} to (2,0);
  \end{pic}
  \] 
We set the remaining generating sets $\MM_{V,W}$ to be empty for all other words $V,W$. This results in our base SMC $\CC$. These \emph{cups} and \emph{caps} ($\eta$ and $\epsilon$) are added to remain faithful to simple UC, where connecting machines in a cycle is allowed. Moreover, the cups and caps let us build morphisms $V\to W$ even when the generating set $\MM_{V,W}$ is empty. However, the ability to connect machines cyclically is rarely used in cryptographic practice, as DAGs of ITMs are usually sufficient. 

\subsection{Protocols, hybrids, and the adversary in UC}

Our starting category $\CC$ yields the associated SMC $\CB$ as in Section~\ref{ssec:abstractadversary}. However, we now need to cut down from $\CB$ to certain sub-SMCs in order to model UC faithfully. We will introduce nested sub-SMCs $ \DB \supseteq  \DR$ of $\CB$. 

In each of these the objects are restricted words in the generating object $A_\leftarrow$ of $\CC$ (instead of words in $\{A_\rightarrow,A_\leftarrow\}$) and morphisms $[B_f,f]$ will have the same restriction on the object $B_f$: intuitively, this is to ensure that states in the resulting category are given by networks of ITMs that are subroutines of the environment, ruling out states that also have machines calling the environment. In other words, every object is of the form $A_\leftarrow^n$ for some $n\in\NN$, and every state on $A_\leftarrow^n$ will be an equivalence class of the form $[A_\leftarrow^m,r \colon I\to A_\leftarrow^n \otimes A_\leftarrow^m]$.  The interpretation is that $m$ is the number of backdoors (which will coincide with the number of machines in $r$) and $n$ the number of 
outgoing (honest) connections. As the other generating object is not referenced from now on, we will simply write $A$ instead of $A_\leftarrow$ and draw wires without annotating them by arrows. Moreover, we will impose the following additional restrictions on all morphisms $[B_f,f]$.
\begin{itemize}
\item For $\DB$, we want to ensure that every machine that is part of a network  has exactly one tape for the adversary. This is achieved by first setting $\cat{D'_{bd}}$  to be the sub-SMC of $\CB$ generated by (\ie the smallest sub-SMC containing) cups and caps and 
all morphisms of the form $[A,f \colon A^m\to A^n \otimes A]$, where $f$ consists of a single ITM with $m$ subroutine-output tapes and $n+1$ input tapes. This implements the restriction to one backdoor tape per machine. However, we also want to ensure that all backdoors and interfaces are oriented similarly, so that resources in the resulting category only expect inputs from the environment (and the adversary) instead of calling them. This is achieved by defining $\DB$ to be the sub-SMC of $\cat{D'_{bd}}$ containing all objects of the form $A^n$ and all morphisms $[A^m,f \colon A^n\to A^{k} \otimes A^m]$. 

\item For $\DR$ we additionally require that all of the allowed ITMs behave in some fixed prescribed  way (\eg giving full control and all of its history to the adversary) when instructed by the adversary. 
For example, in the case of full corruptions, this is achieved by first setting $\cat{D'_{real}}$ to be the sub-SMC of $\CB$ generated by  cups and caps and  all morphisms of the form $[A,f \colon A^m\to A^n \otimes A]$, where $f$ consists of a single ITM with $m$ subroutine-output tapes and $n+1$ input tapes and \emph{moreover}, when instructed to do so at the $(n+1)$st tape, $f$ gives full control and leaks its current state to the machine it is connected to.
We then obtain $\DR$ by defining it to be the sub-SMC of $\cat{D'_{real}}$ consisting of all objects of the form $A^n$ and all morphisms $[A^m,f \colon A^n\to A^{k} \otimes A^m]$ of  $\cat{D'_{real}}$. 
\end{itemize} 
Finer details on corruptions can be layered as additional conventions as in UC. For instance, to ensure that corruptions are (suitably) commensurate between the real and ideal, we might require that the environment learns some fixed function of the set of corrupted identities (\eg which parties are corrupted). 

\subsection{The execution model and indistinguishability}
\label{ssec:indistinguishability}

We now fill in the details of computational indistinguishability.
First, we spell out some further details of ITMs. 
\begin{itemize}
    \item First of all, we allow our ITMs to be probabilistic, which can be modeled by each of them having access to an internal randomness tape, containing an infinite sequence of uniformly random  bits. 
    \item Machines are efficient in a suitable sense. 
    Generally speaking, choosing a good notion of efficiency is surprisingly subtle~\cite{HUM12}. For us, what matters is that there is a notion of efficiency for machines such that any network of efficient machines can be replaced by an efficient machine emulating the entire network. This justifies allowing the adversary to consist of multiple machines. This property is only required for comparison with UC. 
    For Theorem~\ref{thm:translation} to hold, our notion of efficiency needs to also agree with that of UC (for static systems), which in turn also relies on the above.
\end{itemize}

We now define, for each object $X$ of $\CC$, the equivalence relation of computational indistinguishability on states $I\to X$ of $\CC$. First, we define \emph{an environment} of type $X$ to be a morphism $\E \colon X \to A$ in $\CC$, \ie a network of ITMs with one input tape on the right and subroutine-output tapes given by $X$ on the left with some further constraints explained below. 

Now, given a state $r \colon I \to X$ in $\CC$, we feed this state into the environment, resulting in $\E \circ r$:
   \[ \begin{pic}
    \node[state,minimum width=8mm] (r) at (0,0){$r$};
    \node[box=0/1/0/1, minimum height=8mm] (e) at (1,0) {\E};
    \draw (r.east) to (e.west.1);
    \draw (e.east.1) to ++(.5, 0) node[above] {$A$};
  \end{pic}\]
We will now explain how this is to be executed when $\E$ is given an initial input $z$, following the single-threaded execution model of UC.

Initially, the environment $\E$ receives the input $z$ (whose length is interpreted as the security parameter) on the machine to which the only input tape of $\E$ belongs, the \emph{main machine of $\E$}, which is then activated. After that, the machines take turns being executed as follows: the currently active machine is executed until it either sends a message to another machine, after which the message is written on the corresponding tape of the recipient machine (\ie inputs get written on the input tape and outputs get written on the output tape) which now becomes active, or the machine halts without sending a message. If the machine that halted was not the main machine of $\E$, we activate the main machine of $\E$ again. If it was the main machine of  $\E$, the whole execution terminates. We also assume that the main machine of $\E$ never halts without sending a message, whether as input to one of the machines to its left or as a decision bit $b$ written to subroutine-output on its right. Thus the computation goes back and forth between $\E$ and $r$ until $\E$ concludes by returning a bit $b\in\{0,1\}$. 

For any state $r$, input $z$ and $\E$, we thus obtain a probability distribution $\exec_{r,\E}(z)$ on $\{0,1\}$, where the randomness comes from internal randomness used by $\E$ and $r$. We thus obtain a probability ensemble $\{\exec_{r,\E}(z)\}_{z\in\{0,1\}^*}$, which we denote by $\exec_{r,\E}$. We say that $r,s \colon I\to X$ in $\CC$ are \emph{computationally indistinguishable}, denoted $r\approx s$, if for any environment $\E$ of type $X$, the (statistical) difference between the ensembles $\exec_{r,\E}$ and $\exec_{s,\E}$ (\ie its distinguishing advantage) is negligible in $\left\vert z\right\vert$. Note that when we say $r$ and $s$ are indistinguishable this in particular means that $r$ and $s$ have the same codomain (and thus have the same type). Note also that in our notation the adversary (resp., simulator) is absorbed into $r$ (resp., $s$). 

We next observe that computational indistinguishability respects composition \ie satisfies Definition~\ref{def:indistinguishability}. This lemma performs, once and for all, the absorption of parts of the network into the environment, which in particular justifies the practice of not drawing the environment in diagrams.

\begin{lemma}\label{lem:indistinguishability}
Computational indistinguishability respects (sequential and parallel) composition. 
\end{lemma}
\begin{proof}
Let us overload notation as follows: given $r,s \colon I\to X$ and an environment $\E$ of the correct type, we will write 
   \[ \begin{pic}
    \node[state,minimum width=8mm] (r) at (0,0){$r$};
    \node[box=0/1/0/1, minimum height=8mm] (e) at (1,0) {\E};
    \draw (r.east) to (e.west.1);
    \draw (e.east.1) to ++(.6, 0) node[above] {$A$};
  \end{pic}\enspace \approx \enspace
  \begin{pic}
    \node[state,minimum width=8mm] (r) at (0,0){$s$};
    \node[box=0/1/0/1, minimum height=8mm] (e) at (1,0) {\E};
    \draw (r.east) to (e.west.1);
    \draw (e.east.1) to ++(.6, 0) node[above] {$A$};
  \end{pic}\] to mean that the (statistical) difference between the ensembles $\exec_{r,\E}$ and $\exec_{s,\E}$  is negligible. 
With this notation, the proof is just a matter of absorbing parts into the environment, indicated by dashed lines. 

Let $r,s \colon I\to X$ be states in $\CC$, and let $f \colon X\to Y$ be a morphism in $\CC$. Let us assume $r\approx s$. To show that $f\circ r\approx f\circ  s$, let $\E$ be an arbitrary environment and compute as follows:
   \[ \begin{pic}
    \node[state,minimum width=8mm] (r) at (0,0){$r$};
    \node[box=0/1/0/1, minimum height=8mm] (f) at (1,0) {f};
    \node[box=0/1/0/1, minimum height=8mm] (e) at (2,0) {\E};
    \draw (r.east) to (f.west.1);
    \draw (f.east.1) to (e.west.1);
    \draw (e.east.1) to ++(.6, 0) node[above] {};
    \coordinate (extra) at ($(e)+(0.4,0)$);
    \node[fit=(f) (e) (extra), draw, dashed] (bigbox) {};
  \end{pic}\enspace \approx \enspace
  \begin{pic}
    \node[state,minimum width=8mm] (r) at (0,0){$s$};
    \node[box=0/1/0/1, minimum height=8mm] (f) at (1,0) {f};
    \node[box=0/1/0/1, minimum height=8mm] (e) at (2,0) {\E};
    \draw (r.east) to (f.west.1);
    \draw (f.east.1) to (e.west.1);
    \draw (e.east.1) to ++(.6, 0) node[above] {};
    \coordinate (extra) at ($(e)+(0.4,0)$);
    \node[fit=(f) (extra) (e), draw, dashed] (bigbox) {};
  \end{pic}\]
where the (implicit) step of adding a dashed area corresponds 
to $\E\circ f$ being a valid environment and follows from associativity of composition, and the displayed step follows from the assumption that $r\approx s$. 

For parallel composition, let $r_i,s_i \colon I\to X_i$ be states in $\CC$ for $i=1,2$ such that $r_i\approx s_i$. Let $\E$ be arbitrary. We then reason as follows:
\begin{align*} & \begin{pic}
    \node[state,minimum width=8mm] (r1) at (0,1.15) {$r_1$};
    \node[state,minimum width=8mm] (r2) at (0,0) {$r_2$};
    \node[box=0/1/0/2, minimum height=23mm] (e) at (1,.575) {\E};
    \draw (r1.east) to (e.west.1);
    \draw (r2.east) to (e.west.2);
    \draw (e.east.1) to ++(.3, 0) node[above] {};
  \end{pic}  =  
  \begin{pic}
    \node[state,minimum width=8mm] (r1) at (0,1.15) {$r_1$};
    \node[state,minimum width=8mm] (r2) at (.8,0) {$r_2$};
    \node[box=0/1/0/2, minimum height=23mm] (e) at (1.6,.575) {\E};
    \draw (r1.east) to (e.west.1);
    \draw (r2.east) to (e.west.2);
    \draw (e.east.1) to ++(.4, 0) node[above] {};
    \coordinate (extra) at ($(e)+(0.4,0)$);
    \node[fit=(r2) (e) (extra), draw, dashed] (bigbox) {};
  \end{pic}
    \approx  
  \begin{pic}
    \node[state,minimum width=8mm] (r1) at (0,1.15) {$s_1$};
    \node[state,minimum width=8mm] (r2) at (.8,0) {$r_2$};
    \node[box=0/1/0/2, minimum height=23mm] (e) at (1.6,.575) {\E};
    \draw (r1.east) to (e.west.1);
    \draw (r2.east) to (e.west.2);
    \draw (e.east.1) to ++(.4, 0) node[above] {};
    \coordinate (extra) at ($(e)+(0.4,0)$);
    \node[fit=(r2) (e) (extra), draw, dashed] (bigbox) {};
  \end{pic} = \\ 
    &\begin{pic}
    \node[state,minimum width=8mm] (r1) at (0.8,1.15) {$s_1$};
    \node[state,minimum width=8mm] (r2) at (0,0) {$r_2$};
    \node[box=0/1/0/2, minimum height=23mm] (e) at (1.6,.575) {\E};
    \draw (r1.east) to (e.west.1);
    \draw (r2.east) to (e.west.2);
    \draw (e.east.1) to ++(.4, 0) node[above] {};
    \coordinate (extra) at ($(e)+(0.4,0)$);
    \node[fit=(r1) (e) (extra), draw, dashed] (bigbox) {};
  \end{pic}  \approx
    \begin{pic}
    \node[state,minimum width=8mm] (r1) at (0.8,1.15) {$s_1$};
    \node[state,minimum width=8mm] (r2) at (0,0) {$s_2$};
    \node[box=0/1/0/2, minimum height=23mm] (e) at (1.6,.575) {\E};
    \draw (r1.east) to (e.west.1);
    \draw (r2.east) to (e.west.2);
    \draw (e.east.1) to ++(.4, 0) node[above] {};
    \coordinate (extra) at ($(e)+(0.4,0)$);
    \node[fit=(r1) (e) (extra), draw, dashed] (bigbox) {};
  \end{pic}
   = 
   \begin{pic}
    \node[state,minimum width=8mm] (r1) at (0,1.15) {$s_1$};
    \node[state,minimum width=8mm] (r2) at (0,0) {$s_2$};
    \node[box=0/1/0/2, minimum height=23mm] (e) at (1,.575) {\E};
    \draw (r1.east) to (e.west.1);
    \draw (r2.east) to (e.west.2);
    \draw (e.east.1) to ++(.3, 0) node[above] {};
  \end{pic}
  \end{align*}
The first, third and fifth steps follow from axioms of an SMC, the second by $r_1\approx s_1$ the fourth from $r_2\approx s_2$.  It follows that $r_1 \otimes r_2\approx s_1 \otimes s_2$. 
\end{proof}

\begin{remark}\label{rem:environmentasnetwork}
We now observe that nothing so far has formally required that networks of machines can be simulated by a single machine from the same class, and in particular Theorem~\ref{thm:categoricalUCcomposition} goes through without such an assumption.  This contrasts with a point made by Canetti~\cite{Can20}, who notes that ``the UC theorem is, in general, false in settings where systems of ITMs cannot be simulated on a single ITM from the same class.'' In our view, this is an artifact of requiring environments to be single machines:  the composition theorem holds without assuming that a network can be reduced to a single ITM of the same class, provided one allows environments to be networks of machines as they are then automatically closed under composition.    
\end{remark}

We conclude this section with a technical result in $\CC$, used to relate adversaries in the categorical setting to those of UC. 

For each $m$, let $\mux_m \colon A^m\to A$ be the ITM that relays each message arriving on its
$i$th subroutine output tape (on the left) as a subroutine  output to the machine it is connected to on the right, tagged with $i$, and relays each tagged message $(i,\msg)$ it receives on
its single input tape back along the $i$th wire (behaving in some fixed, efficient but otherwise arbitrary way on messages that are not appropriately tagged); let $\demux_m \colon A\to A^m$ be the same relay with the
roles of the inputs and subroutine-output exchanged.

\begin{lemma}[Completeness of the multiplexing adversary]\label{lem:muxcompleteness}
For every $m$ and every state $r \colon I\to Y \otimes A^m$ in $\CC$,
\[
\bigl(\id[Y] \otimes(\demux_m\circ\mux_m)\bigr)\circ r \;\approx\; r .
\]
Consequently, for resources $r,s$ and a morphism $f \colon r\to s$ of $\DR$ whose combined
backdoor is $A^m$, the adversary $\mux_m \colon A^m\to A$ is split mono on $f$ applied to $r$. 
\end{lemma}

\begin{proof}
The round-trip $\demux_m\circ\mux_m$ sends every message out and back under a faithful tagging, hence reproduces it exactly. As a result, it is indistinguishable from $\id[A^m]$ on any state. This is exactly the
hypothesis~\eqref{eq:splitmono} of Corollary~\ref{cor:splitmonocompleteness} with $a=\mux_m$ and $\bar a=\demux_m$.
\end{proof}

In particular, security may always be checked against an adversary whose interface to the environment is the single wire $A$ by Corollary~\ref{cor:splitmonocompleteness}.

\subsection{Relationship to and differences with UC}\label{ssec:translatingtoUC}

Now that we have the subcategories $\DR\subseteq  \DB$ and a composition-respecting equivalence relation, the theory in Section~\ref{sec:UCcategorical} readily applies. In particular, Theorems~\ref{thm:securemapsareanSMC} and~\ref{thm:categoricalUCcomposition} give rise to composition theorems for this specific model. We will now discuss how this model relates to UC. 

Our setting differs in some technical details from that of UC, despite being in essence the same theory. First of all, for us, an ITM can have  multiple communication tapes. In contrast, in UC every machine has only one input tape and one subroutine-output tape (and one backdoor tape).

This, in turn, gives rise to a second difference as to how the connectivity of a network of interactive Turing machines is captured. In our setting, we explicitly add data, keeping track of which tapes are paired together. In contrast, in UC each machine is modeled as a triple $\mu=(\ucid_\mu,\uccom_\mu,\tilde{\mu})$, where $\ucid_\mu$ is the \emph{identity of a machine}, \ie a string giving the machine its name, $\uccom_\mu$ is its \emph{communication set}, \ie a set of pairs of the form $(\ucid, \tp)$, where $\ucid$ is the name of a machine and $\tp \in \{\uci,\ucsro,\ucbd \}$, and finally $\tilde{\mu}$ is the actual program (code) of the machine. 

For a set $\pi$ of such machines  to form \emph{a protocol}, the following conditions must be satisfied.
\begin{itemize}
    \item The identities of the machines in $\pi$ are distinct from each other and from $\{0,1\}$ (as $0$ and $1$ are reserved for the environment and the adversary, respectively).
    \item Backdoor communication is not used,\footnote{Formally, there are no such restrictions on protocols in~\cite[Section 2]{Can20}.  However, we believe this is an oversight. Indeed, Canetti's recent tutorial~\cite{Can25:tutorial} requires this explicitly.} \ie every entry $(\ucid, \tp)$ appearing in the communication set of some machine in $\pi$ has $ \tp \in\{\uci,\ucsro\}$.
    \item If $(\ucid,\uci)$ is in the communication set of some machine $\mu \in \pi$ with identity $\ucid_\mu$, then there is a machine in $\pi$ whose identity is $\ucid$ and whose communication set contains $(\ucid_\mu,\ucsro)$. 
    \item If $(\ucid,\ucsro)$ is in the communication set of some machine $\mu$ in $\pi$ \emph{and} $\ucid$ is the identity of some machine in $\pi$, then that machine's communication set contains $(\ucid_\mu,\uci)$.
\end{itemize}
 In the last condition above, if $\ucid$ is not the identity of any machine in $\pi$, we say that $\mu$ is a \emph{main machine} of $\pi$ and that $\ucid$ is an \emph{external identity} of $\pi$.

Another difference is that in UC, the adversary is always a single machine, whereas for us the adversary is allowed to be a network of machines. This difference is not essential, as by our standing assumptions on efficiency,  we can replace an efficient network of adversaries with an efficient single adversary simulating the network. We argue that the other differences do not matter either, by giving a formal translation between our setting and UC. 

We start by noting that the possible kinds of connectivity do  slightly differ between the two approaches. The main difference is that our formalism allows a machine to be connected by multiple wires to itself, to another machine, or to the environment:
\[
\begin{pic}
\node[box=0/3/0/1, minimum height=10mm] (a) at (0,0) {f};
\node[box=0/1/0/3, minimum height=10mm] (b) at (1.5,.35) {g};
\draw (a.west.1) to ++(-.6, 0) node[above] {};
\draw (a.east.1) to (b.west.2);
\draw (a.east.2) to (b.west.3);
\draw (a.east.3) to ++(2.1, 0) node[below] {};
\draw (b.east.1) to ++(.6, 0) node[below] {};
\draw (b.west.1) to ++(-2.1, 0) node[above] {};
\end{pic} 
\]
We first focus on such connections internal to a resource.

\begin{definition}
Call a resource $r$ in $\DB$ simple if no machine internal to $r$ is connected by more than one wire of the same type to a machine in $r$. 
\end{definition}

We are now in a position to translate from our formalism to UC and back. However, as our machines do not have names (identities) but the ones in UC do, such a naming must be provided.  Conversely, when translating a UC-protocol into a state, we need to specify an ordering (so that we know how to order the outgoing wires in a state). The translation theorem is then cleanest to state in terms of functions operating on resources and protocols equipped with the required additional data, which we next define. 
 
\begin{definition}\label{def:translationdata} Fix a (countably infinite) set  $\names$ of possible identities of machines. Given a representative $(A^m, r \colon I\to A^n \otimes A^m)$ of a resource on  $A^n$ in $\DB$, we say that a function $\sigma \colon \{1, \dots , n\}\to \names$ is \emph{boundary-compatible} with $r$ if no machine of $r$ has two  external wires both mapped to the same identity under $\sigma$. A \emph{valid naming} for $(A^m,r)$ consists of 
\begin{itemize}
    \item a function  $\sigma \colon \{1, \dots , n\}\to \names $  that is boundary-compatible with $r$, and 
    \item a function $\tau \colon \{1,\dots ,m\}\to\names$ that is injective
\end{itemize}
such that the images of $\sigma$ and $\tau$ are disjoint, the image of $\sigma$ does not contain $1$ and the image of $\tau$ does not contain $0,1$. A \emph{named resource} is a representative of a resource together with a valid naming for it. Given a named resource $\tilde{r}$, we will write $\left\vert \tilde{r}\right\vert$ for its underlying state \ie if $\tilde{r} = (r, \sigma, \tau)$, then $\left\vert \tilde{r}\right\vert= r$.

Given a UC-protocol $\pi$, let $\conn(\pi)$ be the set of 
pairs $(\ucid_\mu,\ucid)$ where $\mu$ is a main machine of $\pi$ and $\ucid$ is an external identity of $\pi$ appearing in the communication set of $\mu$, and let $\machines (\pi)$ be the set of machines in $\pi$. An \emph{ordering of the interfaces of $\pi$} consists of total orders on $\conn(\pi)$ and $\machines(\pi)$. An \emph{ordered UC protocol} consists of an UC protocol together with an ordering of its interfaces. Given an ordered UC-protocol $\tilde{\pi}$, we will write $\left\vert\tilde{\pi}\right\vert$ for its underlying protocol. 
\end{definition}

\begin{theorem}[Translation]\label{thm:translation}
There is a translation function $F$ sending any named simple resource  to an ordered UC protocol, and an inverse translation $G$ sending any ordered UC protocol to a named simple resource. The functions $F$ and $G$ satisfy the following:

\begin{enumerate}

\item The translation is well-defined: if $r' = (\id[A^n] \otimes c) \circ r$  for a permutation $c \colon A^m \to A^m$ (witnessing $[A^m,r] = [A^m,r']$  in $\DB$), then $\left\vert F(r', \sigma,\tau\circ c^{-1})\right\vert=\left\vert F(r, \sigma,\tau)\right\vert$.

\item  Round-trips are indistinguishable: $\left\vert G(F(\tilde{r}))\right\vert \approx \left\vert\tilde{r}\right\vert$ as states in  $\CC$, and $\left\vert F(G(\tilde{\pi}))\right\vert\approx \left\vert\tilde{\pi}\right\vert$ as UC protocols.

\item The translation extends to adversaries and environments as follows. A $\CC$-morphism  $a \colon A^m \to A$ (a candidate adversary acting on the backdoor of $r$) translates to a UC adversary $F(a)$ for $F(\tilde{r})$. In particular, $\mux_m$ from Lemma~\ref{lem:muxcompleteness} translates to the UC-dummy adversary and vice versa.  An environment 
$\E \colon A^n \otimes A \to A$ in $\CC$ translates to a UC environment $F(\E)$. Conversely, a UC adversary machine $\mathcal{A}$ attached to a UC hybrid protocol $\pi$ 
translates to a $\CC$-morphism $A^m \to A$ acting on the backdoor of $\left\vert G(\tilde{\pi})\right\vert$; a UC environment compatible with $\left\vert\tilde{\pi}\right\vert$
translates to a $\CC$-morphism $G(\E) \colon A^n \otimes A \to A$. With these translations, probability 
ensembles agree in both directions:
\begin{align*}
    &\exec_{(\id[A^n] \otimes a) \circ r, \E} = \ucexec_{\left\vert F(\tilde{r})\right\vert, F(a), F(\E)} \\
    &\ucexec_{\left\vert\tilde{\pi}\right\vert, \mathcal{A}, \E} = \exec_{(\id[A^n] \otimes G(\mathcal{A})) \circ \left\vert G(\tilde{\pi})\right\vert, G(\E)}~.
\end{align*}
\item $F$ and $G$ preserve computational indistinguishability.

\item Let $r\colon I\to X$ and $s\colon I\to Y$ be resources and $g\colon X\to Y$ a map in $\DR$, such that  $g\circ r$ and $s$ are simple. Assume $(\sigma,\tau)$ and $(\sigma,\tau')$ give valid namings for $g\circ r$ and $s$, so that $F(g \circ  r,\sigma,\tau)$ and $F(s,\sigma,\tau')$ are compatible in that $\conn(F(g \circ  r,\sigma,\tau))=\conn (F(s,\sigma,\tau'))$. Then $g$ is a secure emulation $r \to s$ if and only if  $F(g \circ  r,\sigma,\tau)$ UC-emulates $F(s,\sigma,\tau')$.
\end{enumerate}
\end{theorem}

\begin{proof}
We begin by defining $F(r,\sigma,\tau)$. Let $f_i$ be the machine connected to the $i$th adversarial tape of $r$. Equip $f_i$ with the identity $\tau(i)$, and read off the communication sets from the connectivity of $r$. As $r$ is simple, this is easy: use the names provided above and read the type of connection in $\{\uci,\ucsro \}$ from the object in $\{A_\rightarrow,A_\leftarrow\}$ corresponding to the wire.  We now deal with wires that are not connected to anything: for the $j$th such wire, add $(\sigma(j),\ucsro)$ to the communication set. We also adjust the code of each machine slightly: when a machine gets a message $(\ucid\in \names,\msg)$,  it behaves as the original machine did when getting the message $\msg$ on the wire connected to $\ucid$. The ordering of interfaces of the resulting protocol is inherited from the order on the wires of $r$. 

For $G$, we can read off the connections from the communication sets.  We modify the machines as follows: when a machine gets a message $\msg$ on a wire connected to a machine with identity $\ucid\in \names$, it behaves as the original machine did when receiving message $(\ucid,\msg)$. Moreover, for each machine, we add one more tape connected to the adversarial port: when receiving a message there, the machine will behave as it did when receiving messages via its backdoor tape. 

The ordering of the interfaces of $\pi$ fixes the orders of the honest and backdoor wires. Once those have been fixed, the naming functions $\sigma$ and $\tau$ are inherited from the names of the external identities and the machines of $\pi$, and they will result in a valid naming because $\pi$ is a well-formed protocol.

1) and 2) are now clear from the construction.\footnote{In fact, the round-trips produce networks that are perfectly indistinguishable, but we have not defined that notion here.}

3) Let us first discuss the adversaries. For $F$, we proceed as above, except we first replace the network of adversarial machines by a single (equivalent) machine which we name $1$, and add $(1,\ucbd)$ to the communication sets of all other machines. For $G$, we replace all backdoor connections with a wire of type $A$, and otherwise proceed as above.  Any environment can be translated along $F$ and $G$ similarly, 
and by construction, the resulting probability ensembles are equal.

4) It follows from 3) that if $r$ and $s$ are computationally distinguishable, so are $F(r,\sigma,\tau)$ and $F(s,\sigma,\tau)$ whenever $\sigma,\tau$ give valid namings for both $r$ and $s$, and a similar claim holds for $G$. Hence $F$ and $G$ are injective on indistinguishability classes. That $F$ and $G$ also preserve indistinguishability follows from the following general fact: if $(X,\approx_X)$ and $(Y,\approx_Y)$ are two sets equipped with equivalence relations, and $F \colon X\to Y$ and $G \colon Y\to X$ are functions that are injective on equivalence classes satisfying $G(F(x))\approx_X x$ and $F(G(y))\approx_Y y$ for all $x\in X, y \in Y$, then $F$ and $G$ also preserve the equivalence relations, \ie $x \approx_X y \implies F(x) \approx_Y F(y)$ and similarly for $G$. 

5) We now prove that the security notions are equivalent. Assume first that $g \colon r\to s$ is a secure emulation in $\DR$. Let $\mathcal{A}$ be a UC adversary for $F(g\circ r,\sigma,\tau)$. As $g \colon r\to s$ is secure, for $G(\mathcal{A}) \colon A^m\to A$ there is a simulator such that~\eqref{eq:UCsecurity} holds and hence $F(g \circ  r,\sigma,\tau)$ UC-emulates $F(s,\sigma,\tau')$. Conversely, if $F(g \circ  r,\sigma,\tau)$ UC-emulates $F(s,\sigma,\tau')$, then for the UC dummy adversary $\mathcal{A} $ there is a corresponding simulator $\mathcal{S}$. As $\mathcal{A}$ translates to $\mux_m$, this implies that  $g \colon r\to s$ is a secure emulation by Lemma~\ref{lem:muxcompleteness} and Corollary~\ref{cor:splitmonocompleteness}. 
\end{proof}

Using Theorem~\ref{thm:translation}  we see that Corollary~\ref{cor:splitmonocompleteness} when applied to $\mux_m$ from Lemma~\ref{lem:muxcompleteness} translates to the ``completeness of the dummy adversary'' in UC. One can also use Theorem~\ref{thm:translation} to deduce that any UC-network is indistinguishable from a naming-invariant one---going back and forth via the translation will hard-code the names, so that renaming machines no longer affects their behavior. 

We now discuss the composition theorem of UC for static systems~\cite[Theorem~3]{Can20}.
Let $\pi, \phi, \rho$ be UC protocols. Recall:
\begin{itemize}
    \item $\phi$ is a \emph{subroutine} of $\rho$ if $\phi \subseteq \rho$ as a set of ITMs.
    \item $\pi$ is \emph{compatible} with $\phi$ if their main machines have the same set of identities (\ie there is an identity-preserving \emph{bijective} correspondence between the main machines of $\pi$ and $\phi$), and for each identity in this set, external identities in the communication set of the corresponding machines are equal. Equivalently, $\pi$ is compatible with $\phi$ iff $\conn(\pi)=\conn(\phi)$.
    \item For $\phi$ a subroutine of $\rho$, $\pi$ is \emph{identity-compatible} with $\rho$ and $\phi$ if no machine in $\pi$ has the same identity as one in $\rho \setminus \phi$ or one of the external identities of $\rho$.
\end{itemize}
The first two notions above have natural categorical analogues.  If $\phi$ is a subroutine of $\rho$, then translating them along Theorem~\ref{thm:translation} results in two resources $r_\phi$ and $r_\rho$ such that $r_\rho$ \emph{factorizes via $r_\phi$}, \ie there exists a map $f$ such that $r_\rho=f\circ r_\phi$. Indeed, this $f$ can be built from $\rho\setminus\phi$ by mimicking how UC protocols are translated. 
Two compatible protocols are translated by Theorem~\ref{thm:translation} into two resources \emph{on the same object} (equipped with suitably compatible namings). The notion of identity-compatibility seems to be more specific to UC. Namely, if $\pi$ and $\phi$ are compatible, then their translates $r_\phi$ and $r_\pi$ are resources on the same object, say $X$. If furthermore $\phi$ is a subroutine of $\rho$, we get the factorization $r_\rho=f\circ r_\phi$. We can then readily compose $f$ with $r_\pi$ obtaining a state $f\circ r_\pi$: no identity-compatibility needed. However, for $(\rho\setminus\phi)\cup\pi$ to be a protocol (and hence for the composition to happen prior to translation), we furthermore need to ensure that there are no clashing identities in $\pi$ and $\rho\setminus\phi$. We now make this last claim precise.

\begin{claim} 
\label{claim:prot}
Let $\pi,\phi,\rho$ be UC protocols such that $\phi$ is  a subroutine of $\rho$, $\pi$ is compatible with $\phi$, and $\pi$ is identity-compatible with $\rho$ and $\phi$. Then $\rho^{\phi\to\pi}:=(\rho\setminus\phi)\cup\pi$ is a UC protocol.
\end{claim}

\begin{proof}
Since $\pi,\phi$, and $\rho$ are UC protocols and $\pi$ is identity-compatible with $\rho$ and $\phi$, the identities of all machines in $\rho^{\phi\to\pi}$ are distinct from each other and from $\{0,1\}$. Moreover, $\rho^{\phi\to\pi}$ does not use backdoor communication as $\rho,\pi$ (and $\phi$) do not.

Assume now that $(\ucid,\uci)$ is in the communication set of some machine $\mu \in \rho^{\phi\to\pi}$ with identity $\ucid_\mu$. We need to prove that  there is a machine in $\rho^{\phi\to\pi}$ whose identity is $\ucid$ and whose communication set contains $(\ucid_\mu,\ucsro)$. We split into cases using the decomposition  $ \rho^{\phi\to\pi}=(\rho\setminus\phi)\cup\pi$. If $\mu\in\pi$, then the required machine exists in $\pi\subseteq \rho^{\phi\to\pi}$ since $\pi$ is a UC protocol.  If $\mu\in(\rho\setminus\phi)$, then since $\rho$ is a UC protocol, there is a machine in $\rho$  with identity $\ucid$ and whose communication set contains $(\ucid_\mu,\ucsro)$. Either this machine exists in $\rho\setminus\phi$, or it exists in $\phi$. In the first case, the machine in question exists in $\rho^{\phi\to\pi}$. In the second case, the machine in question is a main machine of $\phi$, so by compatibility of $\phi$ and $\pi$, the bijection from main machines of $\phi$ to those of $\pi$ gives a machine with the same identity in $\pi$ whose communication set contains $(\ucid_\mu,\ucsro)$. In any case, the required machine exists in $\rho^{\phi\to\pi}$.

Assume now that  $(\ucid,\ucsro)$ is in the communication set of some machine $\mu$ in $\rho^{\phi\to\pi}$ \emph{and} $\ucid$ is the identity of some machine $\mu'$ in $\rho^{\phi\to\pi}$. We need to prove that the communication set of $\mu'$ contains $(\ucid_\mu,\uci)$.  We now split into four cases using the decomposition $\rho^{\phi\to\pi}=(\rho\setminus\phi)\cup\pi$:
\begin{itemize}
    \item $\mu\in \pi$, $\mu' \in \pi$: follows since $\pi$ is a UC protocol.
    \item $\mu\in (\rho\setminus\phi)$, $\mu' \in (\rho\setminus\phi)$: follows since $\rho$ is a UC protocol.
    \item $\mu\in \pi$, $\mu' \in (\rho\setminus\phi)$: let $\tilde{\mu}\in\phi$ correspond to $\mu$ under the bijection from  main machines of $\pi$ to those of $\phi$: then compatiblity of $\pi$ and $\phi$ implies that $(\ucid,\ucsro)$ is in the communication set of $\tilde{\mu}$. As $\phi\subseteq \rho$ and $\rho$ is a UC protocol, we conclude that the communication set of $\mu'$ contains $(\ucid_\mu,\uci)$.
    \item $\mu\in (\rho\setminus\phi)$, $\mu' \in \pi$ :
    by identity-compatibility $\ucid$ is not an external identity of $\rho$. Therefore there is a machine $\mu''\in \phi$  with identity $\ucid$ and whose communication set contains $(\ucid_\mu,\uci)$, but this contradicts $\phi$ being a protocol. Therefore, this case can never arise.\qedhere
\end{itemize}
Note that the proof uses the bijection between the main machines of $\pi$ and $\phi$ in both directions. 
\end{proof}

A similar claim is made in~\cite[Section~2]{Can20} but we do not believe it is correct as stated there for the following reasons:
defining compatibility there only requires an ``identity-preserving \emph{injective} correspondence between the main machines of $\pi$ and those of $\phi$.'' Assuming this means an injection from the main machines of $\pi$ to those of $\phi$, this is not enough for the claim, since $\rho\setminus\phi$ might want to call a main machine of $\phi$ that is missing from $\pi$. 
If instead this was intended to mean an identity-preserving injection from the main machines of $\phi$ to those of $\pi$, the end result could still fail to be a UC protocol, for one of the new main machines in $\pi$ might want to send subroutine outputs to a machine in $\rho\setminus\phi$, in which case $(\rho\setminus\phi) \cup \pi$ fails to be a well-formed UC protocol. We note that the tutorial~\cite{Can25:tutorial} likewise defines compatibility via an identity- and external-identity-preserving \emph{bijection}, matching our corrected definition.

Similarly,~\cite[Section~2]{Can20} only requires an \emph{inclusion} of external identities. However, this is not enough for the claim, since a machine in $\rho\setminus\phi$ might want to call a main machine in $\phi$ that is missing from the communication set of the corresponding machine in $\pi$. 
\begin{corollary}[\!\!{\cite[Theorem~3]{Can20}}]\label{cor:UCcomposition}
    Assume that $\pi,\phi,\rho$ are UC protocols, $\phi$ is a subroutine of $\rho$,  and $\pi$ is identity-compatible with $\rho$ and $\phi$. If $\pi$ UC-emulates $\phi$, then UC protocol $\rho^{\phi\to\pi}$ UC-emulates $\rho$.
\end{corollary}
\begin{proof}
    Note that $\pi$ UC-emulating $\phi$ implies their compatibility (recall that UC-emulation relies on the notion of indistinguishability, which presupposes compatibility for execution to be well-defined in both cases). Therefore $\rho^{\phi\to\pi}$ is a protocol by Claim~\ref{claim:prot}. Now order $\conn(\pi)=\conn(\phi)$ and the machines of $\pi$ and $\phi$ arbitrarily. Then extend the order on the machines of $\phi$ to those of $\rho$ so that the machines of $\phi$ come last. Let us then translate $\pi,\phi,\rho$ along Theorem~\ref{thm:translation}. Now the translates of $\pi$ and $\phi$ under $G$ give rise to states $r_\pi,r_\phi \colon I\to X$ and $r_\rho\colon I\to Y$ in $\DB$. Moreover, $\rho \setminus \phi$ gets translated to a morphism $f$ such that $f\circ r_\phi=r_\rho$. Therefore, Theorem~\ref{thm:categoricalUCcomposition} applies, so that $f\circ r_\pi$ emulates $f\circ r_\phi=r_\rho$, which translates back via $F$ to the fact that $\rho^{\phi\to\pi}$ UC-emulates $\rho$. 
\end{proof}

Applying the above blueprint to Theorem~\ref{thm:categoricalUCcompositionglobal} we recover the UCGS theorem of~\cite{BCH+:globalwithUC} for static systems. This, in turn, paves the way for composition in the presence of global subroutines in other UC-like settings, such as Quantum UC~\cite{Unr10}. 

In practice, compatibility can be achieved by having the ideal functionality $s$ (and its translated UC counterpart) ``block'' any mismatched ports. This is a natural choice since the environment is universally quantified over. Moreover, this makes any restrictions explicit in the functionality specification. An analogous observation also applies to Theorem~\ref{thm:categoricalUCcompositionglobal}, where it is once again natural to leave the global functionality unchanged and instead impose the blocking restriction on the functionalities that use it.

%% file: 5-conc.tex
\section{Conclusions and Outlook}\label{sec:conclusion}

In this work we have initiated a categorical study of UC for static systems, giving rigorous diagrammatic proofs of basic results such as the composition theorems. Our theory holds generally and gives a template for producing similar theories where ITMs are replaced by other computational systems. Moreover, categorical thinking fixed some prior oversights and resulted in several beneficial changes to UC (doing away with identities of machines altogether, allowing both the environment and the adversary to consist of multiple machines) while remaining inter-translatable with UC. 

Our work also opens up fertile grounds for future work:
\begin{itemize}
\item Techniques used in~\cite{BK23} can readily be used to develop other variants of the theory. For example, we can develop a version where honest machines are partitioned into $n$ sets ($n$ parties), with all communication across the partitions mediated by a shared functionality. We expect that the resulting categorical theory stands in a similar relationship to the simpler variant of UC in~\cite{DBLP:conf/crypto/CanettiCL15} as our theory here relates to UC for static systems. Similarly, we could readily reason about (abstract) distances up to $\varepsilon>0$, rather than up to indistinguishability. 

\item Can we model full UC similarly? Existing categorical tools are not well-suited for studying a dynamically changing network, making this a difficult question. 
  
\item As mentioned in the introduction, our theory is not an instance of~\cite{BK23} nor vice versa. This suggests  finding a unified categorical framework capturing both. However, we believe that in order to  avoid ``premature generalization'', one should first study categorically other approaches to composability such as
Abstract/Constructive Cryptography~\cite{MR11,Mau11} in general and random systems~\cite{M02} as a specific instance, quantum UC~\cite{Unr10}, reactive simulatability~\cite{PW01,BPW04,BPW07}, the IITM model~\cite{KTR20}, and programming-language-based approaches such as IPDL~\cite{DBLP:journals/pacmpl/GancherSFSM23}, ILC~\cite{LHM19}, state-separating proofs~\cite{BDF+18} and others~\cite{MT13,HS15}. 

\item The long-term goal is to investigate whether recasting various frameworks under a categorical umbrella would facilitate their comparison and pave the way for translating results across them. (Indeed, \cite[Theorem~4.10]{BK23} demonstrates an instance where symmetric monoidal functors preserve security.) This could be achieved via a web of security-preserving functors, thereby establishing when a protocol secure in one framework remains secure in others, and allowing for the derivation of a unified composition theorem.

\end{itemize}
\paragraph*{AI disclosure} This manuscript was proofread and edited with the assistance of large language models and subsequently revised by the authors, who are responsible for its content.

%% file: app-A.tex
\section{Composition Theorems, Abstractly}\label{appendix:UC-composition}

In this section we briefly sketch how our theory fits an abstract and general theory in the style of~\cite{BK22,BK23}. However, the theory we develop here is not an instance of~\cite{BK22,BK23}, nor does~\cite{BK22,BK23} (or even its motivating examples) fit the theory developed here. We will leave the question of finding a categorical foundation capturing both for future work. 

We will assume more categorical knowledge, and in particular the notion of a lax monoidal functor. Recall that a \emph{preorder} is a set equipped with a binary relation $\leq$ that is reflexive and transitive. The category $\cat{PreOrd}$ has preorders as its objects and its morphisms consist of monotone functions: functions $f$ satisfying $x\leq y \Rightarrow f(x)\leq f(y)$. The categorical product of preorders (which takes the cartesian product of the underlying sets) equips $\cat{PreOrd}$ with a symmetric monoidal structure. 

If $\CC$ is a category and $F\colon \CC\to  \cat{PreOrd}$ is a functor, then the \emph{Grothendieck construction of $F$}  is the category $\int F$ whose objects are pairs $(X,r)$, where $X$ is an object of $\CC$ and $r$ is an element of $F(X)$, and morphisms $(X,r)\to (Y,s)$ are given by morphisms $f\colon X\to Y$ in $\CC$ such that $F(f) r\leq s$.  It is known that a symmetric lax monoidal structure on the functor $F$ induces an SMC structure on $\int F$~\cite{moeller:monoidalgrothendieck}.

\begin{definition}\label{def:simulationpreorder} Let $r,s \colon I\to X$ be resources, \ie states in $\DB$. Define $r\leq s$ iff the identity on $X$ gives a secure emulation $r\to s$, \ie iff there exists a simulator $b$ such that 
   \[ \begin{pic}
    \node[state,minimum width=1cm] (r) at (0,0){$r$};
    \draw (r.north east) to ++(.6,0) node[above] {};
    \draw (r.south east) to ++(.6,0) node[below] {};
  \end{pic} \enspace \approx  \enspace \begin{pic}
    \node[state,minimum width=1cm] (r) at (0,0){$s$};
    \node[box=0/1/0/1, minimum height=5mm] (b) at (1.1,-.23) {b};
    \draw (r.north east) to ++(1.7,0) node[above] {};
    \draw (r.south east) to (b.west.1);
    \draw (b.east.1) to ++(.5, 0) node[above] {};
  \end{pic}\]
\end{definition}

Picking $b$ to be the identity shows that $\leq$ is reflexive, and by composing simulators we see that $\leq$ is transitive.  Thus, $\leq$ gives a \emph{preorder} on resources of type $X$. Moreover, the proof of Theorem~\ref{thm:categoricalUCcomposition} shows that if $r\leq s$, then $f\circ r\leq f\circ s$ for any $f \colon X\to Y$ in $\DB$. Therefore, the functor $\hom(I,-) \colon \DB\to\cat{Set}$ can be promoted to a functor $\hom(I,-) \colon \DB\to\cat{PreOrd}$ landing in the category of preordered sets. Moreover, the functor $\hom(I,-)$ has a canonical symmetric lax monoidal structure with respect to the monoidal structures on $\DB$ and $\cat{PreOrd}$. 

We therefore have a composite symmetric lax monoidal functor 
\[\DR\hookrightarrow\DB\xrightarrow{\hom(I,-)} \cat{PreOrd}\]%
which we will simply denote by $F$, so that $\int F$ is an SMC.

\begin{theorem} The SMC $\int F$ is isomorphic to the SMC of resources (states in $\DB$) and secure emulations (maps in $\DR$ satisfying Definition~\ref{def:security}) between them. 
\end{theorem}

\begin{proof}
The result follows by unwinding definitions: an object of $\int F$ consists of pairs $(X,r)$ where $X$ is an object of $\DR$ and $r\colon I\to X$ is a state in $\DB$, \ie a resource. Moreover, morphisms $(X,r)\to (Y,s)$ in $\int F$ are given by morphisms $f\colon X\to Y$ in $\DR$ such that $f\circ r\leq s$. By Definition~\ref{def:simulationpreorder} and Theorem~\ref{thm:categoricalcompletenessofdummy} these are exactly the secure emulations $r\to s$. The monoidal structure on $\int F$ is obtained from the lax monoidal structure on $F$, and it is straightforward but tedious to check that this sends the resources $(X,r)$ and $(Y,s)$ to the resources $(X\otimes Y,r\otimes s)$, and thus ultimately results in the same SMC structure. 
\end{proof}